\documentclass[11pt,british,refpage,intoc,bibliography=totoc,index=totoc,BCOR=7.5mm,captions=tableheading]{extarticle}
\usepackage{mathptmx}
\usepackage[T1]{fontenc}
\usepackage[latin9]{inputenc}
\usepackage[a4paper]{geometry}
\usepackage{color}
\usepackage{babel}
\usepackage{mathrsfs}
\usepackage{amsmath}
\usepackage{amsthm}
\usepackage{amssymb}
\usepackage[numbers]{natbib}
\usepackage[unicode=true,pdfusetitle,
 bookmarks=true,bookmarksnumbered=true,bookmarksopen=false,
 breaklinks=false,pdfborder={0 0 0},pdfborderstyle={},backref=false,colorlinks=true]
 {hyperref}
\hypersetup{
 linkcolor=black, citecolor=blue, urlcolor=black, filecolor=blue, pdfpagelayout=OneColumn, pdfnewwindow=true, pdfstartview=XYZ, plainpages=false}

\makeatletter
\numberwithin{equation}{section}
\numberwithin{figure}{section}
\theoremstyle{plain}
\newtheorem{thm}{\protect\theoremname}[section]
\theoremstyle{definition}
\newtheorem{example}[thm]{\protect\examplename}
\theoremstyle{plain}
\newtheorem{lem}[thm]{\protect\lemmaname}
\theoremstyle{remark}
\newtheorem{rem}[thm]{\protect\remarkname}

\@ifundefined{date}{}{\date{}}
\usepackage{caption}
\usepackage[nottoc]{tocbibind}
\allowdisplaybreaks[4]
\usepackage{dsfont}
\DeclareMathAlphabet{\mathcal}{OMS}{cmsy}{m}{n}

\makeatother

\providecommand{\examplename}{Example}
\providecommand{\lemmaname}{Lemma}
\providecommand{\remarkname}{Remark}
\providecommand{\theoremname}{Theorem}

\begin{document}
\global\long\def\Sgm{\boldsymbol{\Sigma}}%

\global\long\def\W{\boldsymbol{W}}%

\global\long\def\H{\boldsymbol{H}}%

\global\long\def\P{\mathbb{P}}%

\global\long\def\Q{\mathbb{Q}}%

\title{Stochastic formulation of incompressible fluid flows \\ in wall-bounded
regions }
\author{By Zhongmin Qian\thanks{Mathematical Institute, University of Oxford, Oxford OX2 6GG, UK,
and Oxford Suzhou Centre for Advanced Research, Suzhou, P. R. China.
Email: \protect\href{mailto:zhongmin.qian@maths.ox.ac.uk}{zhongmin.qian@maths.ox.ac.uk}}}
\maketitle
\begin{abstract}
In this paper we establish a mathematical framework which may be used
to design Monte-Carlo simulations for a class of time irreversible
dynamic systems, such as incompressible fluid flows, including turbulent
flows in wall-bounded regions, and some other (non-linear) dynamic
systems. Path integral representations for solutions of forward parabolic
equations are obtained, and, in combining with the vorticity transport
equations, probabilistic formulations for solutions of the Navier-Stokes
equations are therefore derived in terms of (forward) McKean-Vlasov
stochastic differential equations (SDEs), which provides us with the
mathematical framework for Monte-Carlo simulations of wall-bounded
turbulent flows. 

\medskip


\medskip

\end{abstract}

\section{Introduction}

It was pointed out that, in a report \citep{vonNeumann1949} by von
Neumann about 70 years ago, turbulence may be understood with the
assistance of powerful electronic computer simulations. With the increasing
computational power it is now possible to obtain information about
turbulent flows via direct numerical simulations (DNS), large-eddy
simulation (LES) and other simulation technologies (cf. \citep{Deardorff 1970,Orszag-Patterson1972,Geurts2003,SenguptaBhaumik2019}
for example). One of the common difficulties for implementing turbulent
simulations (such as the finite element method, the finite difference
method and the finite volume method) lies in the non-linear and non-local
nature of the governing fluid dynamic equations, the Navier-Stokes
equations. According to Kolmogorov's small scale theory of isotropic
turbulence, in order to obtain meaningful simulations of turbulent
flows, the space separation has to be smaller than $\textrm{Re}^{-3/4}L$,
where $\textrm{Re}$ is the Reynolds number and $L$ the typical length.
This requirement for a small mesh in numerical schemes leads to huge
computation hours for simulating turbulence, see \citep[Chapter 7]{Davidson2015}
for a very useful analysis about this aspect. It is thus desirable
to develop Monte-Carlo schemes for simulations of turbulent flows
which may avoid the use of small grids in order to numerically solve
non-linear fluid dynamic equations. 

Monte-Carlo simulations are based on the law of large numbers, which
says the expectation of a random variable can be approximated by the
average of independent samplings. Although convergence rates of Monte-Carlo
methods are in general slow, but the advantage of Monte-Carlo schemes
in most cases lies in their capability of dealing with multivariate
dynamic variables. To implement Monte-Carlo schemes for numerically
calculating solutions of some linear and non-linear partial differential
equations, explicit representations of solutions in terms of some
distributions, or in terms of stochastic differential equations, have
to be established. An archetypal example is the numerical scheme for
computing solutions of Schr\"odinger type partial differential equations,
where the Feynman-Kac formula provides us with a useful representation
to the solutions of the Schr\"dinger equations. 

In this paper we establish a mathematical framework for designing
Monte-Carlo schemes for numerically computing solutions of incompressible
fluid flows in wall-bounded regions. The basic idea in our approach
is to derive an equivalent formulation of the fluid dynamic equations
in terms of McKean-Vlasov type (ordinary) stochastic differential
equations which can be solved numerically. Let us describe the basic
idea and the main technical issues which will be solved in the present
paper.

The basic dynamic variable for a fluid flow in a region $D$ is its
velocity $u(x,t)$ for $x\in D$ and $t\geq0$, a time dependent vector
field. The main scientific question in turbulence is to give a good
description of $u(x,t)$. Mathematically determining the velocity
$u(x,t)$ is equivalent to describing its integral curve $X(\xi,t)$
which are solutions to the following differential equation
\begin{equation}
\frac{\textrm{d}}{\textrm{d}t}X(\xi,t)=u(X(\xi,t),t),\quad X(\xi,0)=\xi.\label{X-1}
\end{equation}
For a viscous fluid flow with viscosity constant $\nu>0$, according
to Taylor \citep{Taylor1921}, one may consider the Brownian trajectories
$X(\xi,t)$ instead which are solutions to the random dynamic system
\begin{equation}
\textrm{d}X(\xi,t)=u(X(\xi,t),t)\textrm{d}t+\sqrt{2\nu}\textrm{d}B(t),\quad X(\xi,0)=\xi,\label{X-2}
\end{equation}
where $B(t)$ is a Brownian motion on some probability space. The
main idea is to rewrite the velocity $u(x,t)$ in terms of the distribution
of the Brownian particles $X(\xi,t)$, substituting it into (\ref{X-2}),
so that (\ref{X-2}) becomes a stochastic differential equation involving
the distribution of its solution. We may call this technique as solving
\emph{the closure problem of Taylor's diffusion}. The closure problem
for incompressible fluid flows freely moving in the whole space (without
boundary) has been solved in a recent paper \citep{QianSuliZhang2022}
(see also for a similar probabilistic representation in \citep{ConstantinIyer2008})
by using an idea from the vortex methods \citep{Chorin 1973,Chorin1980,Cottet and Koumoutsakos 2000,Majda and Bertozzi 2002}.
The most important case of fluid flows in wall-bounded regions will
be solved in the present paper. Let us recall some basic ideas in
the study of the vortex dynamics and the vorticity-velocity formulation,
for details one may refer to \citep{ChorinMarsden2000,Constantin2001a,Cottet and Koumoutsakos 2000,Long 1988,Majda and Bertozzi 2002}.

Consider an incompressible fluid flow with viscosity constant $\nu>0$,
whose velocity $u=(u^{1},u^{2},u^{3})$ and pressure $P$ solve the
Navier-Stokes equations
\begin{equation}
\frac{\partial}{\partial t}u^{i}+(u\cdot\nabla)u^{i}-\nu\Delta u^{i}+\frac{\partial}{\partial x^{i}}P=0\label{i-ns1}
\end{equation}
and
\begin{equation}
\sum_{j=1}^{3}\frac{\partial}{\partial x^{j}}u^{j}=0\label{d-ns1}
\end{equation}
where $i=1,2,3$. The sum in (\ref{d-ns1}) is the divergence of $u$,
$\nabla\cdot u$ at any instance, so the second equation means that
the velocity $u(x,t)$ is divergence free (which can be formulated
in the distribution sense on $\mathbb{R}^{3}$). $u_{0}(x)=u(x,0)$
is called the initial data. The boundary condition at infinity has
to be supplied such as the velocity is bounded at infinity, but let
us ignore this kind of technical issues. It is known that the vorticity
$\omega=\nabla\wedge u$, whose components $\omega^{i}=\varepsilon^{ijk}\frac{\partial}{\partial x^{j}}u^{k}$,
plays a dominate r\^ole in the description of turbulent flows. The
vorticity $\omega$ evolves according to the vorticity transport equations
\begin{equation}
\frac{\partial}{\partial t}\omega^{i}+(u\cdot\nabla)\omega^{i}-\nu\Delta\omega^{i}-\sum_{j=1}^{3}S_{j}^{i}\omega^{j}=0\label{i-v-t-01}
\end{equation}
for $i=1,2,3$, which are obtained by differentiating the Navier-Stokes
equations. Here
\begin{equation}
S_{i}^{j}=\frac{1}{2}\left(\frac{\partial u^{j}}{\partial x^{i}}+\frac{\partial u^{i}}{\partial x^{j}}\right)\label{ij-strain}
\end{equation}
is the symmetric tensor of rate-of-strain. Note that the dynamic variables
$u$, $\omega$ and $S$ depend on $(x,t)$ and are time irreversible.
Our goal is to express the velocity $u(x,t)$ in terms of the distribution
of Taylor's diffusion $X(\xi,t)$ defined by (\ref{X-2}). To achieve
this goal, we observe that $X(\xi,t)$ is a diffusion process with
its infinitesimal generator $L_{u}=\nu\Delta+u\cdot\nabla$, whose
transition probability density $p_{u}(\tau,\xi,t,x)$ is the fundamental
solution to $L_{u}+\frac{\partial}{\partial t}$, cf. \citep{Stroock and Varadhan 1979},
\citep{Ikeda and Watanabe 1989} for example, According to a general
fact from the parabolic theory (cf. \citep{Friedman 1964}), $p_{u}(\tau,\xi,t,x)$
coincides with the fundamental solution to the forward heat operator
$L_{u}^{\star}-\frac{\partial}{\partial t}$. Since $u$ is divergence-free,
so that $L_{u}^{\star}=L_{-u}$, hence the vorticity transport equations
\eqref{i-v-t-01} may be written as Schr\"odinger type equations
\begin{equation}
\left(L_{-u}-\frac{\partial}{\partial t}\right)\omega^{i}+\sum_{j=1}^{3}S_{j}^{i}\omega^{j}=0,\label{i-vt-02}
\end{equation}
and therefore it is possible to express $\omega^{i}$ in terms of
the distribution of the Taylor diffusion, the data $S_{j}^{i}$ and
the initial vorticity $\omega_{0}$. For fluid flows for which the
``potential'' term $\sum_{j=1}^{3}S_{j}^{i}\omega^{j}$ vanish identically
(for $i=1,2,3$); in fact it is the case for 2D fluid flows; the vorticity
transport equations become $\left(L_{-u}-\frac{\partial}{\partial t}\right)\omega^{i}=0$,
thus $\omega^{i}$ can be written as an integral of the initial vorticity
$\omega_{0}=\nabla\wedge u_{0}$ against the fundamental solution
$p_{u}(\tau,\xi,t,x)$. Therefore the closure problem for this case
can be solved by using the Biot-Savart law, cf. \citep{Long 1988,Majda and Bertozzi 2002}.
For turbulent flows, it is typical that the non-linear vorticity stretching
term $\sum_{j=1}^{3}S_{j}^{i}\omega^{j}$ do not vanish. For this
case one may apply the Feynman-Kac formula to \eqref{i-vt-02} in
terms of the time reversed diffusion with generator $L_{-u}$ (but
not $L_{u}$). To resolve the closure problem to \eqref{X-2}, we
need to rewrite the Feynman-Kac formula in terms of the law of the
Taylor diffusion \eqref{X-2}, which will be achieved by using the
duality between the conditional law of the time reversed $L_{-u}$-diffusion
and the conditional law of $L_{u}$-diffusion, cf. \citep{QianSuliZhang2022}
for details.

In this paper we aim to solve the closure problem for incompressible
fluid flows constrained in a domain $D$ with a boundary $\partial D$,
so that schemes may be devised for numerically calculating the solutions,
including the flows within their boundary layers. For a wall-bounded
flow, the velocity $u(x,t)$ satisfies the Navier-Stokes equations
(\ref{i-ns1}, \ref{d-ns1}) only for $x\in D$ and $t>0$, and has
to satisfy the no-slip condition, i.e. $u(x,t)$ vanishes for $x\in\partial D$.
We extend the definition of $u(x,t)$ to all $x\in\mathbb{R}^{3}$
such that $\nabla\cdot u=0$ on $\mathbb{R}^{3}$ in the distribution
sense. This latter requirement is needed in order to ensure again
the duality for the conditional laws of some diffusions involved.
Hence we can define the Taylor diffusion again by (\ref{X-2}) on
the whole space $\mathbb{R}^{3}$.

To solve the closure problem for the boundary problem of the fluid
dynamic equations, we develop several technical tools, mainly the
duality of conditional laws, and the forward type Feynman-Kac formula
for a general class of diffusion processes. We believe that these
results have independent interest by their own and may be useful for
the study of other problems associated with diffusion-reaction equations
(cf. \citep{Freidlin1985,Freidlin1992}). More precisely we will establish
these results for the laws of diffusion processes with infinitesimal
generator being an elliptic operator of second order
\[
\mathscr{L}=\sum_{i,j=1}^{d}\frac{\partial}{\partial x^{j}}a^{ij}(x,t)\frac{\partial}{\partial x^{i}}+\sum_{i=1}^{d}b^{i}(x,t)\frac{\partial}{\partial x^{i}},
\]
although some results are established for more general elliptic operators,
i.e. for elliptic operators not in divergence form.

This paper is organised as the following. In Section 2, we introduce
a class of time dependent elliptic operators of second order which
set up the basic data used throughout the paper. We recall in this
section a few analytic and probabilistic structures, and recall several
well known relations among forward, backward fundamental solutions
and transition probability density functions, scattered in literature,
stated as Lemma \ref{lem3.1}, Lemma \ref{lem4.1} and Lemma \ref{lem4.3}.
In Section 3, by using the classical Feynman-Kac formula, we establish
a general duality relation between the fundamental solution and the
transition probability density function (Lemma \ref{lem5.2} and Lemma
\ref{thm5.3}), and under the divergence free condition on the drift
vector field $b(x,t)$ establish the duality (Lemma \ref{thm5.3-1}).
In Section 4, we prove the main technical tool, a time reverse theorem
(Theorem \ref{theorem3.3}) for the conditional laws of the $\mathscr{L}$-diffusion
when the vector field $b(x,t)$ is divergence free. Thus far we have
established the basic tools for proving our main results. In Section
5, we prove a forward Feynman-Kac formula for solutions of Schr\"odinger
type parabolic systems on the whole space, which generalises the results
in \citep{QianSuliZhang2022}. Section 6 and Section 7 contain the
main contributions of this paper. In Section 6, we establish a forward
Feynman-Kac formula for solutions of the (non homogeneous) boundary
value problems, Theorem \ref{thm7.3}, and in Section 7 we apply the
theory to solve the closure problem for solutions of the Navier-Stokes
equations satisfying the no-slip boundary condition. 

Numerical experiments are not included in the present paper for the
reason that numerical simulations based on the formulation of the
Navier-Stokes equations in Section 7 must be done case by case, and
some of numerical experiments will be published in separate papers. 

In the past decades, many excellent works addressing some probabilistic
aspects of fluid dynamic equations have been published by various
authors, although not from a view-point of numerically calculating
solutions, cf. \citep{Busnello1999}, \citep{Busnello2005}, \citep{Constantin2001a,Constantin2001b},
\citep{ConstantinIyer2008}, \citep{XZhang2010,XZhang2016} for example
and the literature therein. There are still many papers the author
may be not aware of, to the best knowledge of the present author however,
only the Navier-Stokes equations on $\mathbb{R}^{d}$ or $\mathbb{T}^{d}$
($d=2,3$) were considered in a view of stochastic analysis in particular
via the stochastic flow method, and the boundary value problems are
not treated yet in the existing literature. Also in this paper we
adopt an approach of weak solutions both for PDEs and SDEs, and therefore
stochastic flows do not play a r\^ole in our study. This makes distinct
difference between our approach and the existing methods.

\section{Assumptions and notations}

In this section, we introduce several notions and notations which
will be used throughout the paper.

Let $D\subset\mathbb{R}^{d}$ be an open subset of the Euclidean space
of $d$ dimensions, with a smooth boundary $\partial D$, where $d\geq2$.
$\partial D$ is an embedded sub-manifold of dimension $d-1$, though
not necessary connected. $\overline{D}$ denotes the closure of $D$. 

Let $\nu>0$ be a positive constant representing ``the viscosity
constant''. 

Let $a(x,t)=(a^{ij}(x,t))_{i,j\leq d}$ be a Borel measurable, $d\times d$
symmetric matrix-valued function defined for $x\in\mathbb{R}^{d}$
and $t\in\mathbb{R}$. It is assumed that $a(x,t)$ is uniformly elliptic
in the sense that there is a constant $\lambda\geq1$ such that
\begin{equation}
\lambda^{-1}|\xi|^{2}\leq\sum_{i,j=1}^{d}\xi^{i}\xi^{j}a^{ij}(x,t)\leq\lambda|\xi|^{2}\label{u-c01}
\end{equation}
for every $\xi=(\xi^{1},\cdots,\xi^{d})\in\mathbb{R}^{d}$ and $(x,t)\in\mathbb{R}^{d}\times\mathbb{R}$.
It is also assumed that there is a symmetric matrix-valued function
$\sigma(x,t)=(\sigma_{j}^{i}(x,t))$ such that $\sigma_{j}^{i}(x,t)=\sigma_{i}^{j}(x,t)$
and
\begin{equation}
a^{ij}(x,t)=\sum_{l=1}^{d}\sigma_{l}^{i}(x,t)\sigma_{l}^{j}(x,t)\label{a-sigma01}
\end{equation}
for all $i,j\leq d$.

Let $b(x,t)=(b^{1}(x,t),\ldots,b^{d}(x,t))$ be a Borel measurable,
time-dependent vector field on $\mathbb{R}^{d}$, and $c(x,t)$ be
a Borel measurable scalar function on $\mathbb{R}^{d}\times\mathbb{R}$.
$a(x,t)$, $b(x,t)$ and $c(x,t)$ may be defined originally only
for $(x,t)\in D\times J$ (where $J\subset\mathbb{R}$ is an interval),
but without further specification, these functions are automatically
extended to all $(x,t)\in\mathbb{R}^{d}\times\mathbb{R}$ according
to the following rules: $b(x,t)=0$, $c(x,t)=0$ and $a^{ij}(x,t)=\delta^{ij}$
for $(x,t)\notin D\times J$ unless otherwise specified.

Define a time-dependent differential operator of second order
\begin{equation}
L_{a;b,c}=\nu\sum_{i,j=1}^{d}a^{ij}\frac{\partial^{2}}{\partial x^{j}\partial x^{i}}+\sum_{i=1}^{d}b^{i}\frac{\partial}{\partial x^{i}}+c.\label{L-op}
\end{equation}
The first term involving $a^{ij}(x,t)$ on the right-hand side is
called the diffusion part, $b(x,t)$ is called the drift vector field,
and the scalar multiplication part $c(x,t)$ is called the zero-th
order term. If $c=0$, then $L_{a;b,0}$ is a diffusion operator,
which will be denoted by $L_{a;b}$ for simplicity. If $a^{ij}(x,t)=\delta^{ij}$
for all $(x,t)\in\mathbb{R}^{d}\times\mathbb{R}$, then $L_{a;b,c}$
is denoted by $L_{b,c}$, and similarly, if $a^{ij}(x,t)=\delta^{ij}$
and $c(x,t)=0$ for all $(x,t)\in\mathbb{R}^{d}\times\mathbb{R}$,
then $L_{a;b,c}$ is denoted by $L_{b}$. 

$L_{a;b,c}$ may be written in a form whose diffusion part is written
in divergence form
\begin{equation}
L_{a;b,c}=\nu\sum_{i,j=1}^{d}\frac{\partial}{\partial x^{j}}a^{ij}\frac{\partial}{\partial x^{i}}+\sum_{i=1}^{d}\left(b^{i}-\nu\sum_{j=1}^{d}\frac{\partial a^{ij}}{\partial x^{j}}\right)\frac{\partial}{\partial x^{i}}+c,\label{L-op2}
\end{equation}
as long as $\frac{\partial a^{ij}}{\partial x^{j}}$ exist. According
to integration by parts, the formal adjoint $L_{a;b,c}^{\star}$ of
$L_{a;b,c}$ is again a second order differentiable operator and $L_{a;b,c}^{\star}=L_{a;b_{\star},c_{\star}}$,
where
\begin{equation}
b_{\star}^{i}(x,t)=-b^{i}(x,t)+2\nu\sum_{j=1}^{d}\frac{\partial a^{ij}(x,t)}{\partial x^{j}}\label{re-01}
\end{equation}
and
\begin{equation}
c_{\star}(x,t)=c(x,t)-\nabla\cdot b(x,t)+\nu\sum_{i,j=1}^{d}\frac{\partial^{2}a^{ij}(x,t)}{\partial x^{j}\partial x^{i}}\label{re-02}
\end{equation}
respectively. The formal adjoint of the (forward) heat operator $L_{a;b,c}-\frac{\partial}{\partial t}$
is the (backward) heat operator $L_{a;b_{\star},c_{\star}}+\frac{\partial}{\partial t}$. 
\begin{example}
If $a^{ij}=\delta_{ij}$, then $L_{b,c}=\nu\Delta+b\cdot\nabla+c$
and $L_{b,c}^{\star}=L_{-b,c-\nabla\cdot b}$, which is a very important
relation in the study of random vortex methods, cf. \citep{Chorin 1973,Cottet and Koumoutsakos 2000,Majda and Bertozzi 2002}
for details.
\end{example}

Similarly
\begin{equation}
\mathscr{L}_{a;b,c}=\nu\sum_{i,j=1}^{d}\frac{\partial}{\partial x^{j}}a^{ij}\frac{\partial}{\partial x^{i}}+\sum_{i=1}^{d}b^{i}\frac{\partial}{\partial x^{i}}+c\label{div-oper-01}
\end{equation}
is an elliptic operator of second order, and $\mathscr{L}_{a;b,c}=L_{a;\hat{b},c}$,
where
\begin{equation}
\hat{b}^{i}=b^{i}+\nu\sum_{j=1}^{d}\frac{\partial a^{ij}}{\partial x^{j}}\label{div-03}
\end{equation}
as long as $\frac{\partial a^{ij}}{\partial x^{j}}$ exist. Note that
$\hat{b}$ is independent of $c$.

In general Einstein's convention that repeated indices are summed
up in their ranges is applied, unless said otherwise. The following
convention will also be applied to quantities which rely on $a$,
$b$ and $c$. If a quantity depends on $a,b$ and $c$ then it may
be labelled with a lower subscript $a;b,c$. If $a^{ij}(x,t)=\delta_{ij}$
for all $x$ and $t$, then the part $a;$ will be omitted, so that
\begin{equation}
L_{b,c}=\nu\Delta+\sum_{i=1}^{d}b^{i}\frac{\partial}{\partial x^{i}}+c.\label{si-01}
\end{equation}
If $c=0$, then the part $,c$ will be omitted, hence
\begin{equation}
L_{a;b}=\nu\sum_{i,j=1}^{d}a^{ij}\frac{\partial^{2}}{\partial x^{i}\partial x^{j}}+\sum_{i=1}^{d}b^{i}\frac{\partial}{\partial x^{i}}\label{si-02}
\end{equation}
and therefore
\begin{equation}
L_{b}=\nu\Delta+\sum_{i=1}^{d}b^{i}\frac{\partial}{\partial x^{i}}.\label{si-03}
\end{equation}
This convention will be applied to other quantities such as the fundamental
solutions, transition probability density functions, and so on. In
particular $\mathscr{L}_{a;b}=L_{a;\hat{b}}$.

We finally make some comments on further regularity assumptions on
the data $a,b,c$ and $q$ (to be introduced later on) in additional
to the assumptions we have already made. For simplicity, we assume
that all data $a,b,c$ and $q$ are smooth and bounded, although these
regularity conditions may be too demanding for applications. In fact,
in additional to the assumptions made, the following assumptions are
sufficient for our arguments to be true.

1) The results and arguments about $L_{a;b}$ are valid if $a^{ij}(x,t)$
are uniformly continuous in $(x,t)$, $b^{i}(x,t)$, $c(x,t)$ and
$q(x,t)$ are bounded, Borel measurable and continuous in $t$, although
the boundedness conditions may be weaken. 

2) The results and arguments about $\mathscr{L}_{a;b}$ are valid
if $a^{ij}(x,t)$ are Borel measurable and continuous in $t$, $b^{i}(x,t)$,
$c(x,t)$ and $q(x,t)$ are bounded, Borel measurable and continuous
in $t$, although the boundedness conditions may be weaken.

\subsection{Analytic structures}

In this section we recall several analytic structures associated with
the basic data $a(x,t)$, $b(x,t)$ and $c(x,t)$. 

The fundamental solutions are the basic analytic quantities associated
with the elliptic operators of second order, and the reader may refer
to \citep{Aronson 1968} and \citep{Friedman 1964,Ladyzenskaja Solonnikov and Uralceva 1968}
for the most fundamental results in this aspect. 

$\Gamma_{a;b,c}(x,t;\xi,\tau)$, defined for $x,\xi\in\mathbb{R}^{d}$
and $\tau<t$, is the fundamental solution of the \emph{forward} parabolic
equation
\begin{equation}
\left(L_{a;b,c}-\frac{\partial}{\partial t}\right)f(x,t)=0\label{f-parabolic-01}
\end{equation}
in the sense that for every fixed $\tau\geq0$ and $\xi\in\mathbb{R}^{d}$,
$f(x,t)=\Gamma_{a;b,c}(x,t;\xi,\tau)$ solves \eqref{f-parabolic-01},
and for every bounded continuous function $\varphi$ on $\mathbb{R}^{d}$
\begin{equation}
\lim_{t\downarrow\tau}\int_{\mathbb{R}^{d}}\Gamma_{a;b,c}(x,t;\xi,\tau)\varphi(\xi)\textrm{d}\xi=\varphi(x)\label{w-lim-01}
\end{equation}
for every $x\in\mathbb{R}^{d}$. 

Similarly $\Gamma_{a;b,c}^{\star}(x,t;\xi,\tau)$ defined for $x,\xi\in\mathbb{R}^{d}$
and for $t<\tau$ is a fundamental solution of the backward heat operator
$L_{a;b,c}+\frac{\partial}{\partial t}$, if for any $\xi\in\mathbb{R}^{d}$
and $\tau$, $f(x,t)=\Gamma_{a;b,c}^{\star}(x,t;\xi,\tau)$ satisfies
the backward parabolic equation 
\begin{equation}
\left(L_{a;b,c}+\frac{\partial}{\partial t}\right)f(x,t)=0\label{back-p-01}
\end{equation}
and
\[
\lim_{t\uparrow\tau}\int_{\mathbb{R}^{d}}\Gamma_{a;b,c}^{\star}(x,t;\xi,\tau)\varphi(\xi)\textrm{d}\xi=\varphi(x)
\]
for every continuous function $\varphi$ on $\mathbb{R}^{d}$ and
for every $x\in\mathbb{R}^{d}$.

\begin{lem}
\label{lem3.1}The following relation holds:
\begin{equation}
\Gamma_{a;b,c}(x,t;\xi,\tau)=\Gamma_{a;b_{\star},c_{\star}}^{\star}(\xi,\tau;x,t)\label{dual-01-1}
\end{equation}
for any $\tau<t$ and $x,\xi\in\mathbb{R}^{d}$. 
\end{lem}

For a proof of this basic fact, see \citep[Theorem 15, page 28]{Friedman 1964}.

\subsection{Probabilistic structures}

In this section we recall several probabilistic structures associated
with the elliptic operator of second order $L_{a;b}$. If $a^{ij}(x,t)$
are uniformly continuous and $b^{i}(x,t)$ are bounded and Borel measurable,
then according to \citep{Stroock-Varadhan1969} there is a unique
family of probability measures $\mathbb{P}_{a;b}^{\xi,\tau}$ (where
$\tau\geq0$ and $\xi\in\mathbb{R}^{d}$) on the path space $\varOmega=C([0,\infty),\mathbb{R}^{d})$
of all continuous paths in $\mathbb{R}^{d}$, equipped with its Borel
$\sigma$-algebra, such that

1) the diffusion starts from $\xi$ at time $\tau$ in the following
sense that
\begin{equation}
\mathbb{P}_{a;b}^{\xi,\tau}\left[\psi\in\varOmega:\psi(t)=\xi\textrm{ for all }0\leq t\leq\tau\right]=1.\label{mart-01}
\end{equation}

2) (local martingale property) for every $f\in C^{2,1}(\mathbb{R}^{d}\times[\tau,\infty))$,
\[
M_{t}^{[f]}=f(\psi(t),t)-f(\psi(\tau),\tau)-\int_{\tau}^{t}\left(L_{a;b}+\frac{\partial}{\partial s}\right)f(\psi(s),s)\textrm{d}s
\]
for $t\geq\tau\geq0$ and $M_{t}^{[f]}=0$ for $0\leq t\leq\tau$,
is a local martingale under the probability measure $\mathbb{P}_{a;b}^{\xi,\tau}$. 

The family $\mathbb{P}_{a;b}^{\xi,\tau}$ (where $\xi\in\mathbb{R}^{d},\tau\geq0$)
of probability measures is simply called the $L_{a;b}$-diffusion. 

If $a^{ij}(x,t)=\sigma_{k}^{i}(x,t)\sigma_{j}^{k}(x,t)$ and $\sigma_{j}^{i}(x,t)$
and $b^{i}(x,t)$ are Lipschitz continuous, then the $L_{a;b}$-diffusion
may be constructed by solving It\^o's stochastic differential equations:
\begin{equation}
\textrm{d}X^{i}=b^{i}(X,t)\textrm{d}t+\sqrt{2\nu}\sum_{k=1}^{d}\sigma_{k}^{i}(X,t)\textrm{d}B^{k},\;X_{s}=\xi\textrm{ for }s\leq\tau\label{eq:sde-m1}
\end{equation}
for $i=1,\ldots,d$, where $B=(B^{1},\cdots,B^{d})$ is the standard
Brownian motion in $\mathbb{R}^{d}$. The distribution of the strong
solution $X$ is the probability measure $\mathbb{P}_{a;b}^{\xi,\tau}$.

Let $P_{a;b}(\tau,\xi,t,\textrm{d}x)=\mathbb{P}_{a;b}^{\xi,\tau}\left[\psi\in\varOmega:\psi(t)\in\textrm{d}x\right]$
(where $t>\tau\geq0$), called the transition probability function
of the $L_{a;b}$-diffusion. Since $a(x,t)$ is uniformly elliptic,
$P_{a;b}(\tau,\xi,t,\textrm{d}x)$ has a probability density function
$p_{a;b}(\tau,\xi,t,x)$ with respect to the Lebesgue measure, so
that
\[
P_{a;b}(\tau,\xi,t,\textrm{d}x)=p_{a;b}(\tau,\xi,t,x)\textrm{d}x\quad\textrm{ for }t>\tau\geq0\textrm{ and }\xi\in\mathbb{R}^{d}.
\]
The transition density function $p_{a;b}(\tau,\xi,t,x)$ (for $t>\tau\geq0$)
associated with the $L_{a;b}$-diffusion is positive and H\"older
continuous (cf. \citep{Aronson 1968}).
\begin{lem}
\label{lem4.1}Let $p_{a;b}(\tau,\xi,t,x)$ be the transition probability
density function of the $L_{a;b}$-diffusion. Then
\[
p_{a;b}(\tau,\xi,t,x)=\Gamma_{a;b}^{\star}(\xi,\tau;x,t)=\Gamma_{a;\tilde{b},\tilde{c}}(x,t;\xi,\tau)
\]
for all $0\leq\tau<t$ and $x,\xi\in\mathbb{R}^{d}$, where
\begin{equation}
\tilde{b}^{i}(x,t)=-b^{i}(x,t)+2\nu\sum_{j=1}^{d}\frac{\partial a^{ij}(x,t)}{\partial x^{j}}\label{re-01-1}
\end{equation}
and
\begin{equation}
\tilde{c}(x,t)=-\nabla\cdot b(x,t)+\nu\sum_{i,j=1}^{d}\frac{\partial^{2}a^{ij}(x,t)}{\partial x^{j}\partial x^{i}}.\label{re-02-1}
\end{equation}
\end{lem}

This follows immediately from Lemma \ref{lem3.1} and the well known
fact that $p_{a;b}(\tau,\xi,t,x)=\Gamma_{a;b}^{\star}(\xi,\tau;x,t)$
for all $0\leq\tau<t$ and $x,\xi\in\mathbb{R}^{d}$ (cf. \citep{Stroock and Varadhan 1979}). 
\begin{rem}
If $a^{ij}(x,t)=\delta^{ij}$ and $b(x,t)$ is divergence free for
every $t$, that is, $\nabla\cdot b=0$ in the distribution sense,
then 
\begin{equation}
p_{b}(\tau,\xi,t,x)=\Gamma_{b}^{\star}(\xi,\tau;x,t)=\Gamma_{-b}(x,t;\xi,\tau)\label{eq:fund-de-01}
\end{equation}
for all $0\leq t<\tau$. In particular, if $\nabla\cdot b=0$, then
\begin{equation}
\Gamma_{b}(x,t;\xi,\tau)=p_{-b}(\tau,\xi,t,x)\label{eq:fund-de-02}
\end{equation}
for all $t\geq\tau\geq0$. The relationship (\ref{eq:fund-de-02})
plays the crucial r\^ole in the random vortex method (cf. \citep{Long 1988,Majda and Bertozzi 2002}).
\end{rem}

On the other hand, it is known that a forward parabolic equation can
be solved by running back the time from a future time, which gives
rise to another set of relations between fundamental solutions and
the transition probability density functions. If $T>0$ and $f(x,t)$
is a function, then define $f_{T}(x,t)=f(x,(T-t)^{+})$. 
\begin{lem}
\label{lem4.3}Let $T>0$. Then
\begin{equation}
p_{a_{T};b_{T}}(T-t,x,T-\tau,\xi)=\Gamma_{a;b}(x,t;\xi,\tau)=\Gamma_{a;\tilde{b},\tilde{c}}^{\star}(\xi,\tau;x,t)\label{eq:diff-backT1}
\end{equation}
for all $0\leq\tau<t\leq T$ and $x,\xi\in\mathbb{R}^{n}$, where
$\tilde{b}$ and $\tilde{c}$ are given by \eqref{re-01-1} and \eqref{re-02-1}
respectively.
\end{lem}

\begin{proof}
Let $\Theta(x,t;\xi,\tau)=p_{a_{T};b_{T}}(T-t,x,T-\tau,\xi)$ for
$0\leq\tau<t\leq T$ and $\xi,x\in\mathbb{R}^{d}$. As a function
of $(x,t)$, $p_{a_{T};b_{T}}(t,x;\tau,\xi)$ (for $t<\tau$) solves
the backward parabolic equation, so that $\Theta$ solves the forward
parabolic equation
\begin{equation}
\left(L_{a_{T};b_{T}}-\frac{\partial}{\partial t}\right)\Theta(x,t;\xi,\tau)=0\label{eq:diff-bac-01-1}
\end{equation}
for $0\leq\tau<t\leq T$. Since $b_{T}(x,T-t)=b(x,t)$ and $a_{T}(x,T-t)=a(x,t)$
for $0\leq t\leq T$, the previous equality is equivalent to the forward
parabolic equation: 
\[
\left(L_{a;b}-\frac{\partial}{\partial t}\right)\Theta(x,t;\xi,\tau)=0
\]
on $\mathbb{R}^{d}\times[\tau,T]$ for every $\tau\in(0,T]$. Suppose
$\varphi(\xi)$ is continuous and bounded on $\mathbb{R}^{d}$, then
\begin{align*}
\int_{\mathbb{R}^{d}}\Theta(x,t;\xi,\tau)\varphi(\xi)\textrm{d}\xi & =\int_{\mathbb{R}^{d}}p_{a_{T};b_{T}}(T-t,x,T-\tau,\xi)\varphi(\xi)\textrm{d}\xi\\
 & \rightarrow\varphi(x)
\end{align*}
as $T-t\uparrow T-\tau$, i.e. as $t\downarrow\tau$. By the uniqueness
of the fundamental solution $\Theta(x,t;\xi,\tau)$ coincides with
$\Gamma_{a;b}(x,t;\xi,\tau)$, and therefore the conclusion follows
immediately. 
\end{proof}
\begin{rem}
For the case where $a^{ij}(x,t)=\delta^{ij}$ and $b(x,t)$ is divergence-free,
then $\tilde{b}=-b$ and $\tilde{c}=0$, so that 
\begin{equation}
p_{b^{T}}(T-t,x,T-\tau,\xi)=\Gamma_{b}(x,t;\xi,\tau)=\Gamma_{-b}^{\star}(\xi,\tau;x,t)=p_{-b}(\tau,\xi,t,x)\label{eq:rev-01}
\end{equation}
for all $0\leq\tau<t\leq T$ and $x,\xi\in\mathbb{R}^{n}$. Therefore
for this case the fundamental solution $\Gamma_{-b}^{\star}(\xi,\tau;x,t)$
is a probability density in $x$ and $\xi$ respectively.
\end{rem}

In general by definition $\Gamma_{a;\tilde{b},\tilde{c}}^{\star}(\xi,\tau;x,t)$
is a probability density in the variable $\xi$, while in general
it is not a probability density function with respect to the variable
$x$, thus can not be a transition probability density function of
an diffusion. A sufficient condition to ensure that $\Gamma_{a;\tilde{b},\tilde{c}}^{\star}(\xi,\tau;x,t)$
is a transition probability function (with respect to $x$) of some
diffusion is given in the next section.

\section{The Feynman-Kac formula}

The Feynman-Kac formula is a functional integration representation
to the solutions $f^{i}(x,t)$ of a backward parabolic equation:
\begin{equation}
\left(L_{a;b}+\frac{\partial}{\partial t}\right)f^{i}(x,t)+\sum_{j=1}^{n}q_{j}^{i}(x,t)f^{j}(x,t)=0\quad\textrm{ in }\mathbb{R}^{d}\times[0,\infty),\label{bp-fey-01}
\end{equation}
where $i=1,\ldots,n$, and $n$ is some positive integer. 

For each $\psi\in C([0,\infty),\mathbb{R}^{d})$, let $Q(\tau,\psi,t)=(Q_{j}^{i}(\tau,\psi,t))_{i,j\leq n}$
be the solution to the ordinary differential equations:
\begin{equation}
\frac{\textrm{d}}{\textrm{d}t}Q_{j}^{i}(t)=Q_{k}^{i}(t)q_{j}^{k}(\psi(t),t),\quad Q_{j}^{i}(\tau)=\delta_{j}^{i}\label{gu-01}
\end{equation}
for $t\geq\tau$, if the solution exists and is unique.
\begin{lem}
\label{lem5.1}(The Feynman-Kac formula) Suppose $a^{ij}(x,t)$ are
uniformly continuous, $b^{i}(x,t)$ are bounded and Borel measurable,
and $q_{j}^{i}(x,t)$ are bounded and continuous. Suppose $f^{i}(x,t)$
are $C^{2,1}$ solutions to \eqref{bp-fey-01} with bounded first
derivatives. Then
\begin{equation}
f^{i}(\eta,\tau)=\int_{\varOmega}Q_{j}^{i}(\tau,\psi;t)f^{j}(\psi(t),t)\mathbb{P}_{a;b}^{\eta,\tau}(\textrm{d}\psi)\label{b-fey-s1}
\end{equation}
for $t\geq\tau$ and $\eta\in\mathbb{R}^{d}$.
\end{lem}

As a consequence we have the following
\begin{lem}
\label{lem5.2} Under the same regularity assumptions on $a(x,t)$
and $b(x,t)$ as in Lemma \ref{lem5.1}, and suppose $c(x,t)$ is
bounded and continuous. Then
\begin{equation}
\Gamma_{a;b,c}^{\star}(\xi,\tau;x,t)=p_{a;b}(\tau,\xi,t,x)\int_{\varOmega}C(\tau,\psi;t)\mathbb{P}_{a;b}^{\xi,\tau\rightarrow x,t}(\textrm{d}\psi)\label{rep-g1}
\end{equation}
for $x,\xi\in\mathbb{R}^{d}$ and $t>\tau\geq0$, where $C(\tau,\psi;r)$
for each $\psi\in C([0,\infty);\mathbb{R}^{d})$ is the unique solution
to the ordinary differential equation:
\begin{equation}
C(\tau,\psi;t)=1+\int_{\tau}^{\tau\wedge t}C(\tau,\psi;s)c(\psi(s),s)\textrm{d}s\label{gu-04}
\end{equation}
for $t\geq0$. 
\end{lem}

\begin{proof}
Let $T>0$ and $f$ be a $C^{2,1}$ solution to the backward parabolic
equation:
\[
\left(L_{a;b,c}+\frac{\partial}{\partial t}\right)f=0
\]
such that $f(x,t)\rightarrow f_{0}(x)$ as $t\uparrow T$. The previous
equation may be rewritten as the following:
\[
\left(L_{a;b}+\frac{\partial}{\partial t}\right)f(x,t)+c(x,t)f(x,t)=0
\]
so that, according to Lemma \ref{lem5.1}
\[
f(\xi,\tau)=\int_{\varOmega}C(\tau,\psi;T)f(\psi(T),T)\mathbb{P}_{a;b}^{\xi,\tau}(\textrm{d}\psi)
\]
for $0\leq\tau<T$ , and therefore
\[
f(\xi,\tau)=\int_{\mathbb{R}^{d}}\left(p_{a;b}(\tau,\xi,T,x)\int_{\varOmega}C(\tau,\psi;T)\mathbb{P}_{a;b}^{\xi,\tau}\left[\left.\textrm{d}\psi\right|\psi(T)=x\right]\right)f(x,T)\textrm{d}x
\]
which yields that
\[
\Gamma_{a;b,c}^{\star}(\xi,\tau;x,T)=p_{a;b}(\tau,\xi,T,x)\int_{\varOmega}C(\tau,\psi;T)\mathbb{P}_{a;b}^{\xi,\tau}\left[\left.\textrm{d}\psi\right|\psi(T)=x\right].
\]
The proof is complete. 
\end{proof}
\begin{lem}
\label{thm5.3}Let $T>0$. Then
\[
p_{a_{T},b_{T}}(T-t,x,T-\tau,\xi)=p_{a;\tilde{b}}(\tau,\xi,t,x)\int_{\varOmega}\tilde{C}(\tau,\psi;t)\mathbb{P}_{a;\tilde{b}}^{\xi,\tau\rightarrow x,t}(\textrm{d}\psi)
\]
for any $0\leq\tau<t\leq T$ and $x,\xi\in\mathbb{R}^{d}$, where
$\tilde{b}$ and $\tilde{c}$ are given by (\ref{re-01-1}) and (\ref{re-02-1})
respectively, and 
\[
\tilde{C}(\tau,\psi;t)=1+\int_{\tau}^{\tau\wedge t}\tilde{C}(\tau,\psi;s)\tilde{c}(\psi(s),s)\textrm{d}s
\]
for $t\geq0$ and $\psi\in C([0,\infty);\mathbb{R}^{d})$.
\end{lem}

\begin{proof}
We have
\[
p_{a_{T},b_{T}}(T-t,x,T-\tau,\xi)=\Gamma_{a;\tilde{b},\tilde{c}}^{\star}(\xi,\tau;x,t)
\]
for $t>\tau\geq0$, so the corollary follows immediately from Lemma
\ref{lem5.2}.
\end{proof}
We may apply Lemma \ref{thm5.3} to the $\mathscr{L}_{a;b}$-diffusion,
whose transition probability density function is denoted by $h_{a;b}(t,x,\tau,\xi)$
for $0\leq t<\tau$. The formal adjoint operator of $\mathscr{L}_{a;b}$
is given by 
\[
\mathscr{L}_{a;b}^{\star}=\nu\sum_{i,j=1}^{d}\frac{\partial}{\partial x^{j}}a^{ij}\frac{\partial}{\partial x^{i}}-\sum_{i=1}^{d}b^{i}\frac{\partial}{\partial x^{i}}-\nabla\cdot b
\]
and therefore $\tilde{c}=\nabla\cdot b$. Therefore we have the following
consequence.
\begin{lem}
\label{thm5.3-1}Suppose $\nabla\cdot b(\cdot,t)=0$ in the distribution
sense for every $t\geq0$. Then for every $T>0$
\begin{equation}
h_{a_{T};b_{T}}(T-t,x,T-\tau,\xi)=h_{a;-b}(\tau,\xi,t,x)\label{div-04}
\end{equation}
for $0\leq\tau<t\leq T$ and $\xi,x\in\mathbb{R}^{d}$.
\end{lem}

\begin{proof}
Apply Lemma 3.3 to the differential operator $\mathscr{L}_{a;b}=L_{a;\hat{b}}$.
Since $\nabla\cdot b=0$ identically for every $t$, $\tilde{c}=0$,
and therefore the gauge process $\tilde{C}\equiv1$, the equality
follows immediately. 
\end{proof}

\section{Diffusion bridges}

In this section we establish a duality among the conditional laws
of the $\mathscr{L}_{a;b}$-diffusions. Let $\mathbb{Q}_{a;b}^{\eta,\tau}$
be the distribution of the $\mathscr{L}_{a;b}$-diffusion started
from $\eta$ at $\tau\geq0$. If $T>0$ then $\mathbb{Q}_{a;b}^{\eta,0\rightarrow\zeta,T}$
denotes the conditional law of the $\mathscr{L}_{a;b}$-diffusion
started from $\eta$ at instance $0$ and arrived at $\zeta$ at time
$T$. 
\begin{thm}
\label{theorem3.3} Suppose $b(x,t)$ is divergence free and bounded,
then for every $T>0$
\begin{equation}
\mathbb{Q}_{a;-b}^{\eta,0\rightarrow\zeta,T}\circ\tau_{T}=\mathbb{Q}_{a_{T};b_{T}}^{\zeta,0\rightarrow\eta,T}\label{c-duality}
\end{equation}
where $\tau_{T}$ is the time reverse which sends $\psi$ to $\tau_{T}\psi$,
for $\psi\in C([0,T];\mathbb{R}^{d})$, that is, $\tau_{T}\psi(t)=\psi(T-t)$
for $t\in[0,T]$.
\end{thm}

\begin{proof}
It is known that the conditional law of $\mathbb{Q}_{a;b}^{\eta,0}$
given $\psi(T)=\zeta$, where $\eta$ and $\zeta$ are the initial
and final points, according to (14.1) in \citep{Dellacherie and Meyer Volume D},
denoted by $\mathbb{Q}_{a;b}^{\eta,0\rightarrow\zeta,T}$, is also
Markovian (time non-homogeneous) whose transition density function
is given by
\begin{equation}
q(\tau,\xi,t,x)=\frac{h_{a;b}(\tau,\xi,t,x)h_{a;b}(t,x,T,\zeta)}{h_{a;b}(\tau,\xi,T,\zeta)}\label{eq:condit-e1}
\end{equation}
that is
\[
q(\tau,\xi,t,x)=\mathbb{Q}_{a;b}^{\eta,0\rightarrow\zeta,T}\left[\psi\in\varOmega:\psi(t)\in\textrm{d}x|\psi(\tau)=\xi\right]
\]
for $0<\tau<t\leq T$. Let $0=t_{0}<t_{1}<\cdots<t_{n}<t_{n+1}=T$.
Then
\[
\mathbb{Q}_{a;b}^{\eta,0\rightarrow\zeta,T}\left[w_{t_{0}}\in dx_{0},w_{t_{1}}\in dx_{1},\cdots,w_{t_{n}}\in dx_{n},w_{t_{n+1}}\in dx_{n+1}\right]
\]
equals the measure
\[
q(0,\eta,t_{1},x_{1})\cdots q(t_{i-1},x_{i-1},t_{i},x_{i})\cdots q(t_{n},x_{n},T,\zeta)dx_{1}\cdots dx_{n}.
\]
By using (\ref{eq:condit-e1}), this measure has a pdf
\[
\frac{h_{a;b}(0,\eta,t_{1},x_{1})\cdots h_{a;b}(t_{i-1},x_{i-1},t_{i},x_{i})\cdots h_{a;b}(t_{n},x_{n},T,\zeta)}{h_{a;b}(0,\eta,T,\zeta)}
\]
with respect to the measure $\delta_{\eta}(dx_{0})dx_{1}\cdots dx_{n}\delta_{\zeta}(dx_{n+1})$.
According to Lemma \ref{thm5.3-1}, the previous pdf equals
\[
\frac{h_{a_{T};-b_{T}}(0,\zeta,T-t_{n},x_{n})\cdots h_{a_{T};-b_{T}}(T-t_{i},x_{i},T-t_{i-1},x_{i-1})\cdots p_{a_{T};-b_{T}}(T-t_{1},x_{1},T,\eta)}{p_{a_{T};-b_{T}}(0,\zeta,T,\eta)}
\]
and we therefore conclude that
\[
\mathbb{Q}_{a;b}^{\eta,0\rightarrow\zeta,T}\left[w_{t_{0}}\in dx_{0},w_{t_{1}}\in dx_{1},\cdots,w_{t_{n}}\in dx_{n},w_{t_{n+1}}\in dx_{n+1}\right]
\]
coincides with
\[
\mathbb{Q}_{a_{T};-b_{T}}^{\zeta,0\rightarrow\eta,T}\left[w_{T-t_{n}}\in dx_{n},\cdots,w_{T-t_{1}}\in dx_{1},w_{T}\in dx_{0},w_{0}\in dx_{n+1}\right].
\]
\end{proof}

\section{Feynman-Kac formula for forward parabolic equation}

Let us begin with the following general construction of the gauge
functional. 

Given $q(x,t)=(q^{ij}(x,t))_{i,j\leq n}$,  an $n\times n$ square-matrix
valued function defined for $x\in\mathbb{R}^{d}$ for $t\geq0$. For
each $T>0$ and each continuous path $\psi\in C([0,T];\mathbb{R}^{d})$,
consider the following ordinary differential equation
\begin{equation}
\frac{\textrm{d}}{\textrm{d}t}Q_{j}^{i}(t)=\sum_{k=1}^{n}Q_{k}^{i}(t)q_{j}^{k}(\psi(t),(T-t)^{+}),\quad Q_{j}^{i}(0)=\delta_{ij}\label{ode-01}
\end{equation}
which is linear in $Q$. The solution depends on $T$ and on $\psi$
as well, so it is denoted by $Q(\psi,T;t)$. By definition
\begin{equation}
Q_{j}^{i}(\psi,T;t)=\delta_{ij}+\int_{0}^{t}\sum_{k=1}^{n}Q_{k}^{i}(\psi,T;s)q_{j}^{k}(\psi(s),(T-s)^{+})\textrm{d}s\label{ode-02}
\end{equation}
for $t\in[0,T]$, where $i,j=1,\cdots,n$.

The time reversal operator $\tau_{T}$, we recall, on $C([0,T];\mathbb{R}^{d})$
maps $\psi$ to $\tau_{T}\psi(t)=\psi(T-t)$. Therefore, if we substitute
$\psi$ by $\tau_{T}\psi$ and $t$ by $T-t$ we obtain
\begin{align*}
Q_{j}^{i}(\tau_{T}\psi,T;T-t) & =\delta_{ij}+\int_{0}^{T-t}\sum_{k=1}^{n}Q_{k}^{i}(\tau_{T}\psi,T;s)q_{j}^{k}(\psi(T-s),(T-s)^{+})\textrm{d}s\\
 & =\delta_{ij}+\int_{t}^{T}\sum_{k=1}^{n}Q_{k}^{i}(\tau_{T}\psi,T;T-s)q_{j}^{k}(\psi(s),s)\textrm{d}s.
\end{align*}
Hence we have the following elementary fact.
\begin{lem}
For each $\psi\in C([0,T];\mathbb{R}^{d})$, $Q(\psi,T;t)$ denotes
the solution to ODE (\ref{ode-01}) and $G(\psi,T;t)=Q(\tau_{T}\psi,T;T-t)$
for $t\in[0,T]$. Then $G$ is the solution to the ordinary differential
equation
\begin{equation}
\frac{d}{dt}G_{j}^{i}(t)=-\sum_{k=1}^{n}G_{k}^{i}(t)q_{j}^{k}(\psi(t),t),\quad G_{j}^{i}(T)=I.\label{back-ode-01}
\end{equation}
Moreover $Q(\psi,T;t)=G(\tau_{T}\psi,T;T-t)$ for $t\in[0,T]$. 
\end{lem}

Now we are in a position to study the initial value problem of the
forward parabolic equation:
\begin{equation}
\left(L-\frac{\partial}{\partial t}\right)w^{i}(x,t)+\sum_{j=1}^{n}q_{j}^{i}(x,t)w^{j}(x,t)+f^{i}(x,t)=0\label{01-t}
\end{equation}
subject to the initial value $w^{i}(x,0)=w_{0}^{i}(x)$, where $i=1,\ldots,n$,
where $L=L_{a;b}$ (for this case we assume that $a^{ij}(x,t)$ are
uniformly continuous) or $L=\mathscr{L}_{a;b}$, for both cases, $b^{i}(x,t)$
are bounded and Borel measurable. 
\begin{lem}
\label{lem7.2}Let $\mathbb{P}^{\eta}$ denote the distribution of
the $L_{a_{T};b_{T}}$-diffusion or the $\mathscr{L}_{a_{T};b_{T}}$-diffusion,
depending on whether $L=L_{a;b}$ or $L=\mathscr{L}_{a;b}$, started
from $\eta\in\mathbb{R}^{d}$ at time $0$. Then
\begin{equation}
w^{i}(x,T)=\mathbb{P}^{x}\left[\sum_{j=1}^{n}Q_{j}^{i}(\psi,T;T)w_{0}^{j}(\psi(T))+\int_{0}^{T}Q_{j}^{i}(\psi,T;t)f^{j}(\psi(t),T-t)\textrm{d}t\right].\label{back-03-1}
\end{equation}
\end{lem}

\begin{proof}
This is the classical Feynman-Kac formula applying to $w(x,T-t)$.
\end{proof}
\begin{thm}
\label{thm7.3}Suppose $b$ is divergence free on $\mathbb{R}^{d}$,
i.e. $\nabla\cdot b=0$ in the distribution sense. Suppose $w(x,t)=(w^{i}(x,t))_{i\leq n}$
is a solution to the parabolic system
\begin{equation}
\left(\mathscr{L}_{a;b}-\frac{\partial}{\partial t}\right)w^{i}(x,t)+\sum_{j=1}^{n}q_{j}^{i}(x,t)w^{j}(x,t)+f^{i}(x,t)=0\label{01-t-1}
\end{equation}
subject to the initial value $w(x,0)=w_{0}(x)$, where $i=1,\ldots,n$.
Then 
\begin{align}
w^{i}(x,T) & =\sum_{j=1}^{n}\int_{\mathbb{R}^{d}}\left(\int_{\varOmega}G_{j}^{i}(\psi,T;0)\mathbb{Q}_{a;-b}^{\xi,0\rightarrow x,T}(\textrm{d}\psi)\right)w_{0}^{j}(\xi)h_{a;-b}(0,\xi,T,x)\textrm{d}\xi\nonumber \\
 & +\sum_{j=1}^{n}\int_{\mathbb{R}^{d}}\left(\int_{\varOmega}\int_{0}^{T}G_{j}^{i}(\psi,T;t)f^{j}(\psi(t),t)\textrm{d}t\mathbb{Q}_{a;-b}^{\xi,0\rightarrow x,T}(\textrm{d}\psi)\right)h_{a;-b}(0,\xi,T,x)\textrm{d}\xi.\label{rep-041}
\end{align}
where $G(\psi,T;t)=G(t)$ (for every $T>0$ and every continuous path
$\psi$) is the unique solution to the ordinary differential equation:
\begin{equation}
\frac{\textrm{d}}{\textrm{d}t}G_{j}^{i}(t)=-G_{k}^{i}(t)q_{j}^{k}(\psi(t),t),\quad G_{j}^{i}(T)=\delta_{ij}.\label{back-ode51}
\end{equation}
\end{thm}

\begin{proof}
According to the previous Lemma \ref{lem7.2}
\begin{align*}
w^{i}(x,T) & =\sum_{j=1}^{n}\mathbb{Q}^{x,0}\left[Q_{j}^{i}(\psi,T;T)w_{0}^{j}(\psi(T))\right]+\int_{0}^{T}\mathbb{Q}^{x,0}\left[Q_{j}^{i}(\psi,T;t)f^{j}(\psi(t),T-t)\right]\textrm{d}t\\
 & =\sum_{j=1}^{n}\int_{\mathbb{R}^{d}}\mathbb{Q}^{x,0}\left[\left.Q_{j}^{i}(\psi,T;T)w_{0}^{j}(\psi(T))\right|\psi(T)=\xi\right]\mathbb{Q}^{x,0}\left[\psi(T)\in d\xi\right]\\
 & +\int_{0}^{T}\mathbb{Q}^{x,0}\left[\left.Q_{j}^{i}(\psi,T;t)f^{j}(\psi(t),T-t)\right|\psi(T)=\xi\right]\mathbb{Q}^{x,0}\left[\psi(T)\in d\xi\right]\textrm{d}t\\
 & =\sum_{j=1}^{n}\int_{\mathbb{R}^{d}}\mathbb{Q}_{a_{T};b_{T}}^{x,0\rightarrow\xi,T}\left[Q_{j}^{i}(\psi,T;T)\right]w_{0}^{j}(\xi)h_{a_{T};b_{T}}(0,x,T,\xi)\textrm{d}\xi\\
 & +\sum_{j=1}^{n}\int_{0}^{T}\mathbb{Q}_{a_{T};b_{T}}^{x,0\rightarrow\xi,T}\left[Q_{j}^{i}(\psi,T;t)f^{j}(\psi(t),T-t)\right]h_{a_{T};b_{T}}(0,x,T,\xi)d\xi\textrm{d}t.
\end{align*}
Since $b(x,t)$ is divergence-free, so that $\mathscr{L}_{a;b}^{\star}=\mathscr{L}_{a;-b}$.
By Theorem \ref{theorem3.3}
\[
\mathbb{Q}_{a_{T};b_{T}}^{\xi,0\rightarrow\eta,T}\circ\tau_{T}=\mathbb{Q}_{a;-b}^{\eta,0\rightarrow\xi,T},
\]
together with the relation that $Q(\psi,T;t)=G(\tau_{T}\psi,T;T-t)$,
we obtain
\begin{align*}
w^{i}(x,T) & =\sum_{j=1}^{n}\int_{\mathbb{R}^{d}}\mathbb{Q}_{a_{T};b_{T}}^{x,0\rightarrow\xi,T}\left[G_{j}^{i}(\tau_{T}\psi,T;0)\right]w_{0}^{j}(\xi)h_{a_{T};b_{T}}(0,x,T,\xi)d\xi\\
 & +\int_{0}^{T}\mathbb{Q}_{a_{T};b_{T}}^{x,0\rightarrow\xi,T}\left[G_{j}^{i}(\tau_{T}\psi,T;T-t)f^{j}(\tau_{T}\psi(T-t),T-t)\right]h_{a_{T};b_{T}}(0,x,T,\xi)d\xi\textrm{d}t\\
 & =\sum_{j=1}^{n}\int_{\mathbb{R}^{d}}\mathbb{Q}_{a;-b}^{\xi,0\rightarrow x,T}\left[G_{j}^{i}(\psi,T;0)\right]w_{0}^{j}(\xi)h_{a_{T};b_{T}}(0,x,T,\xi)d\xi\\
 & +\sum_{j=1}^{n}\mathbb{Q}_{a;-b}^{\xi,0\rightarrow x,T}\left[\int_{0}^{T}G_{j}^{i}(\psi,T;T-t)f^{j}(\psi(T-t),T-t)\textrm{d}t\right]h_{a_{T};b_{T}}(0,x,T,\xi)d\xi.
\end{align*}
Finally according to Lemma \ref{thm5.3-1}
\[
h_{a_{T};b_{T}}(T-t,x,T-\tau,\xi)=h_{a;-b}(\tau,\xi,t,x),
\]
which yields that $h_{a_{T};b_{T}}(0,x,T,\xi)=h_{a;-b}(0,\xi,T,x)$
for $0\leq\tau<t\leq T$. Substituting this equation into the representation
for $w^{i}$ we obtain \eqref{rep-041}.
\end{proof}

\section{Boundary value problems}

We first recall the Feynman-Kac formula for solutions of the initial
and boundary value problem:
\begin{equation}
\left(L_{a;b}+\frac{\partial}{\partial t}\right)f^{j}(x,t)+q_{k}^{j}(x,t)f^{k}(x,t)=g^{j}(x,t)\quad\textrm{ in }D\times[0,\infty)\label{D-01}
\end{equation}
subject to the Dirichlet boundary condition that 
\begin{equation}
f^{j}(x,t)=\beta^{j}(x)\quad\textrm{ for }x\in\partial D\textrm{ and }t>0,\label{D-02}
\end{equation}
where $j=1,\cdots,n$. 
\begin{lem}
\label{lem7.1} Assume that $a$ is uniformly continuous and $b$
is bounded and Borel measurable. For each $\psi\in C([0,\infty);\mathbb{R}^{d})$
denote $Q(\psi,t)$ the solution to the (linear) ordinary differential
equation
\begin{equation}
\frac{\textrm{d}}{\textrm{d}t}Q_{j}^{i}(t)=Q_{k}^{i}(t)1_{D}(\psi(t))q_{j}^{k}(\psi(t),t),\quad Q_{j}^{i}(0)=\delta_{j}^{i}\label{pr-Q-eq-01}
\end{equation}
and 
\begin{equation}
\zeta_{D}(\psi)=\inf\{t\geq0:\psi(t)\notin D\}.\label{rep-0r1}
\end{equation}
Then the following integration representation holds:
\begin{align}
f^{i}(\eta,0) & =\int_{\varOmega}Q_{j}^{i}(\psi,t)f^{j}(\psi(t),t)1_{\left\{ \zeta_{D}(\psi)>t\right\} }\mathbb{P}^{\eta}(\textrm{d}\psi)\nonumber \\
 & +\int_{\varOmega}Q_{j}^{i}(\psi,\zeta_{D}(\psi))\beta^{j}(\psi(\zeta_{D}(\psi)))1_{\left\{ \zeta_{D}(\psi)\leq t\right\} }\mathbb{P}^{\eta}(\textrm{d}\psi)\\
 & -\int_{\varOmega}\left[\int_{0}^{t\wedge\zeta_{D}(\psi)}Q_{j}^{i}(\psi,s)g^{j}(\psi(s),s)\textrm{d}s\right]\mathbb{P}^{\eta}(\textrm{d}\psi)\label{fey-f-02}
\end{align}
for all $t\geq0$ and $\eta\in\mathbb{R}^{d}$, where $i=1,\cdots,n$
and $\mathbb{P}^{\eta}=\mathbb{P}_{a;b}^{\eta,0}$ for simplicity.
\end{lem}

\begin{proof}
For a slightly different version of Feynman-Kac formula and its proof,
see for example \citep[Theorem 2.3. page133]{Freidlin1985}. For completeness
we include a proof here. We may assume that $q_{j}^{i}(x,t)=0$ for
$x\notin D$ otherwise we may use $1_{D}(x)q_{j}^{i}(x,t)$ instead.
Let $X$ be the weak solution of the stochastic differential equation
\[
\textrm{d}X^{k}(t)=b^{k}(X(t),t)\textrm{d}t+\sqrt{2\nu}\sigma_{l}^{k}(X(t),t)\textrm{d}B^{l}(t).
\]
for $k=1,\ldots,d$, with initial $X(0)=\eta$. Consider the linear
ordinary differential equation:
\begin{equation}
\frac{\textrm{d}}{\textrm{d}t}Q_{j}^{i}(t)=Q_{k}^{i}(t)q_{j}^{k}(X(t),t),\quad Q_{j}^{i}(0)=\delta_{j}^{i}.\label{Q-bd-01}
\end{equation}
Suppose $\rho$ is smooth such that $\rho(x)=1$ for $x\in\overline{D}$,
and $\tilde{f}^{i}(x,t)=\rho(x)f^{i}(x,t)$. Then $\tilde{f}^{i}(x,t)$
($i=1,\ldots,n$) are $C^{2,1}$-functions on $\mathbb{R}^{d}\times[0,\infty)$
and
\begin{equation}
\frac{\partial}{\partial t}\tilde{f}^{j}+L_{a;b}\tilde{f}^{j}+q_{k}^{j}\tilde{f}^{k}=F^{j}\quad\textrm{ in }\mathbb{R}^{d}\times[0,\infty).\label{par-q-5}
\end{equation}
 Consider $M_{t}^{i}=Q_{j}^{i}(t)f^{j}(X(t),t)$ for $t\geq0$. Then,
according to It\^o's formula
\begin{align}
M_{t}^{i} & =M_{0}^{i}+\int_{0}^{t}Q_{j}^{i}(s)\sqrt{2\nu}\sigma_{l}^{k}(X(s),s)\frac{\partial\tilde{f}^{j}}{\partial x^{k}}(X(s),s)\textrm{d}B_{s}^{l}\nonumber \\
 & +\int_{0}^{t}Q_{j}^{i}(s)\left(\frac{\partial}{\partial s}\tilde{f}^{j}+L_{a;b}\tilde{f}^{j}+q_{k}^{j}\tilde{f}^{k}\right)(X(s),s)\textrm{d}s.\label{mart-bQ-04}
\end{align}
Let 
\[
\zeta_{D}=\inf\{t\geq0:X(t)\notin D\}
\]
be the first time the diffusion leaves the region $D$. Then
\begin{equation}
\mathbb{E}\left[M_{t\wedge\zeta_{D}}^{i}\right]=\mathbb{E}\left[M_{0}^{i}\right]+\mathbb{E}\int_{0}^{t\wedge\zeta_{D}}Q_{j}^{i}(s)\left(\frac{\partial}{\partial s}f^{j}+L_{a;b}f^{j}+q_{k}^{j}f^{k}\right)(X(s),s)\textrm{d}s.\label{eq:mart-E-01}
\end{equation}
Since $f^{j}$ solve the differential equations \eqref{par-q-5},
so that 
\begin{equation}
\mathbb{E}\left[M_{t\wedge\zeta_{D}}^{i}\right]=\mathbb{E}\left[M_{0}^{i}\right]+\mathbb{E}\left[\int_{0}^{t\wedge\zeta_{D}}Q_{j}^{i}(s)F^{j}(X(s),s)\textrm{d}s\right].\label{mart-E-2}
\end{equation}
Since $M_{0}^{i}=f^{i}(\eta,0)$ and 
\begin{align*}
\mathbb{E}\left[M_{t\wedge\zeta_{D}}^{i}\right] & =\mathbb{E}\left[M_{\zeta_{D}}^{i}:t\geq\zeta_{D}\right]+\mathbb{E}\left[M_{t}^{i}:t<\zeta_{D}\right]\\
 & =\mathbb{E}\left[Q_{j}^{i}(\zeta_{D})f^{j}(X(\zeta_{D}),\zeta_{D}):t\geq\zeta_{D}\right]\\
 & +\mathbb{E}\left[Q_{j}^{i}(t)f^{j}(X(t),t):t<\zeta_{D}\right]\\
 & =\mathbb{E}\left[Q_{j}^{i}(t)f^{j}(X(t),t):t<\zeta_{D}\right]\\
 & +\mathbb{E}\left[Q_{j}^{i}(\zeta_{D})\beta^{j}(X(\zeta_{D})):t\geq\zeta_{D}\right]
\end{align*}
where the last equality follows from the Dirichlet boundary condition:
$X(\zeta_{D})\in\partial D$ on $\zeta_{D}<\infty$, and $f^{j}(x,t)=\beta^{j}(x)$
for $x\in\partial D$. Substituting this equality into \eqref{mart-E-2},
\begin{align*}
\mathbb{E}\left[Q_{j}^{i}(t)f^{j}(X(t),t):t<\zeta_{D}\right] & =f^{i}(\eta,0)+\mathbb{E}\left[\int_{0}^{t\wedge\zeta_{D}}Q_{j}^{i}(s)F^{j}(X(s),s)\textrm{d}s\right]\\
 & -\mathbb{E}\left[Q_{j}^{i}(\zeta_{D})\beta^{j}(X(\zeta_{D})):t\geq\zeta_{D}\right]
\end{align*}
The functional integration representation follows by an approximating
procedure. 
\end{proof}
We next establish a forward Feynman-Kac formula. 

For every $\psi\in C([0,\infty);\mathbb{R}^{d})$ and $T>0$, $\tilde{Q}(\psi,T;t)$
denotes the solution to the following linear ordinary differential
equations
\begin{equation}
\frac{\textrm{d}}{\textrm{d}t}\tilde{Q}_{j}^{i}(\psi,T;t)=-\tilde{Q}_{k}^{i}(\psi,T;t)1_{D}(\psi(t))q_{j}^{k}(\psi(t),t),\quad\tilde{Q}_{j}^{i}(\psi,T;T)=\delta_{j}^{i}\label{t-Q-01}
\end{equation}
for $i,j=1,\cdots,n$.
\begin{thm}
\label{thm7.2} Suppose $b(x,t)$ is bounded, Borel measurable and
$\nabla\cdot b=0$ in the distribution sense on $\mathbb{R}^{d}$.
Let $w(x,t)$ be the solution to Cauchy's initial and boundary problem
of the following parabolic system:
\begin{equation}
\left(\mathscr{L}_{a;b}-\frac{\partial}{\partial t}\right)w^{j}(x,t)+\sum_{k=1}^{n}q_{k}^{j}(x,t)w^{k}(x,t)=g^{j}(x,t)\quad\textrm{ in }D\label{sy-01}
\end{equation}
subject to the initial and boundary conditions that
\begin{equation}
w^{j}(x,0)=w_{0}^{j}(x)\textrm{ for }x\in D,\textrm{ and }w^{j}(x,t)=\beta^{j}(x)\textrm{ for }x\in\partial D\textrm{, }t>0\label{int-bd-c3}
\end{equation}
for $j=1,\cdots,n$. Then
\begin{align}
w^{i}(\eta,T) & =\int_{D}\left(\int_{\varOmega}\tilde{Q}_{j}^{i}(\psi,T;0)1_{\left\{ \zeta_{D}(\psi)>T\right\} }\mathbb{Q}_{a;-b}^{\xi,0\rightarrow\eta,T}(\textrm{d}\psi)\right)w_{0}^{j}(\xi)h(0,\xi,T,\eta)\textrm{d}\xi\nonumber \\
 & +\int_{\mathbb{R}^{d}}\left(\int_{\varOmega}\tilde{Q}_{j}^{i}(\psi,\lambda_{T,\partial D}(\psi))\beta^{j}(\psi(\lambda_{T,\partial D}(\psi)))1_{\left\{ \zeta_{D}(\psi)\leq T\right\} }\mathbb{Q}_{a;-b}^{\xi,0\rightarrow\eta,T}(\textrm{d}\psi)\right)h(0,\xi,T,\eta)\textrm{d}\xi\\
 & -\int_{0}^{T}\left[\int_{\mathbb{R}^{d}}\left(\int_{\varOmega}\tilde{Q}_{j}^{i}(\psi,T;s)g^{j}(\psi(s),s)1_{\left\{ \zeta_{D}(\theta_{s}\psi)>T-s\right\} }\mathbb{Q}_{a;-b}^{\xi,0\rightarrow\eta,T}(\textrm{d}\psi)\right)h(0,\xi,T,\eta)\textrm{d}\xi\right]\textrm{d}s\label{fey-mg1}
\end{align}
 for every $\eta\in D$ and $T>0$, where $h(\tau,\xi,t,\eta)$ denotes
$h_{a;-b}(\tau,\xi,t,\eta)$ for simplicity,
\[
\lambda_{T,\partial D}(\psi)=\sup\left\{ t\in[0,T]:\psi(t)\in\partial D\right\} 
\]
and 
\[
\theta_{s}:\varOmega\rightarrow\varOmega,\quad\theta_{s}\psi(t)=\psi(t+s)\;\textrm{ for }s;t\geq0.
\]
\end{thm}

\begin{proof}
Let $T>0$ and $f^{i}(x,t)=w^{i}(x,T-t)$, so that $u^{i}$ satisfy
the following parabolic equations
\[
\left(\mathscr{L}_{a_{T};b_{T}}+\frac{\partial}{\partial t}\right)f^{j}(x,t)+\sum_{k=1}^{n}q_{k,T}^{j}(x,t)f^{k}(x,t)=g_{T}^{j}(x,t)\quad\textrm{ in }D.
\]
Then according to \eqref{fey-f-02}
\begin{align}
w^{i}(\eta,T) & =\int_{\varOmega}\left(Q_{j}^{i}(\psi,T)w_{0}^{j}(\psi(T))1_{\left\{ \zeta_{D}(\psi)>T\right\} }\right)\mathbb{Q}_{a_{T};b_{T}}^{\eta}(\textrm{d}\psi)\nonumber \\
 & +\int_{\varOmega}Q_{j}^{i}(\psi,\zeta_{D}(\psi))\beta^{j}(\psi(\zeta_{D}(\psi)))1_{\left\{ \zeta_{D}(\psi)\leq T\right\} }\mathbb{Q}_{a_{T};b_{T}}^{\eta}(\textrm{d}\psi)\\
 & -\int_{0}^{T}\left[\int_{\varOmega}\left(Q_{j}^{i}(\psi,s)g^{j}(\psi(s),T-s)1_{\left\{ s<T\wedge\zeta_{D}(\psi)\right\} }\right)\mathbb{Q}_{a_{T};b_{T}}^{\eta}(\textrm{d}\psi)\right]\textrm{d}s\label{fey-m}
\end{align}
where $\mathbb{Q}_{a_{T};b_{T}}^{\eta}$ is the law of the $\mathscr{L}_{a_{T};b_{T}}$-diffusion
started from $\eta$ at instance $0$, and 
\[
Q_{j}^{i}(\psi,t)=\delta_{j}^{i}+\int_{0}^{t}Q_{k}^{i}(\psi,s)1_{D}(\psi(s))q_{j}^{k}(\psi(s),T-s)\textrm{d}s.
\]
Replace $t$ by $T-t$ and $\psi$ by $\tau_{T}\psi$ one obtains
that
\begin{align*}
Q_{j}^{i}(\tau_{T}\psi,T-t) & =\delta_{j}^{i}+\int_{0}^{T-t}Q_{k}^{i}(\tau_{T}\psi,s)1_{D}(\psi(T-s))q_{j}^{k}(\psi(T-s),T-s)\textrm{d}s\\
 & =\delta_{j}^{i}-\int_{T}^{t}Q_{k}^{i}(\tau_{T}\psi,T-s)1_{D}(\psi(s))q_{j}^{k}(\psi(s),s)\textrm{d}s.
\end{align*}
Hence, by setting $\tilde{Q}(\psi,T;t)=Q_{j}^{i}(\tau_{T}\psi,T-t)$,
one deduce that
\begin{equation}
\tilde{Q}_{j}^{i}(\psi,T;t)=\delta_{j}^{i}-\int_{T}^{t}\tilde{Q}_{k}^{i}(\psi,T;s)1_{D}(\psi(s))q_{j}^{k}(\psi(s),s)\textrm{d}s.\label{re-flow-01}
\end{equation}
Moreover $Q_{j}^{i}(\tau_{T}\psi,t)=\tilde{Q}(\psi,T;T-t)$ for every
$t\in[0,T]$. 

We rewrite \eqref{fey-m} by conditioning on the values of the diffusion
at $T$, to obtain that
\begin{equation}
w^{i}(\eta,T)=R_{I}^{i}(\eta,T)+R_{B}^{i}(\eta,T)-R_{N}^{i}(\eta,T)\label{Rep-cas-01}
\end{equation}
where the first term
\[
R_{I}^{i}(\eta,T)=\int_{D}\left(\int_{\varOmega}Q_{j}^{i}(\psi,T)1_{\left\{ \zeta_{D}(\psi)>T\right\} }\mathbb{Q}_{a_{T};b_{T}}^{\eta,0\rightarrow\xi,T}(\textrm{d}\psi)\right)w_{0}^{j}(\xi)h_{a_{T};b_{T}}(0,\eta,T,\xi)\textrm{d}\xi,
\]
the boundary term
\[
R_{B}^{i}(\eta,T)=\int_{\varOmega}Q_{j}^{i}(\psi,\zeta_{D}(\psi))\beta^{j}(\psi(\zeta_{D}(\psi)))1_{\left\{ \zeta_{D}(\psi)\leq T\right\} }\mathbb{Q}_{a_{T};b_{T}}^{\eta}(\textrm{d}\psi)
\]
and finally the inhomogeneous term
\[
R_{N}^{i}(\eta,T)=\int_{0}^{T}\left[\int_{\mathbb{R}^{d}}\left(\int_{\varOmega}Q_{j}^{i}(\psi,s)g^{j}(\psi(s),T-s)1_{\left\{ s<\zeta_{D}(\psi)\right\} }\mathbb{Q}_{a_{T};b_{T}}^{\eta,0\rightarrow\xi,T}(\textrm{d}\psi)\right)h_{a_{T};b_{T}}(0,\eta,T,\xi)\textrm{d}\xi\right]\textrm{d}s.
\]

Since $\mathscr{L}_{a;b}^{\star}=\mathscr{L}_{a;-b}$ under our assumptions,
$\mathbb{Q}_{a;-b}^{\xi,0\rightarrow\eta,T}\circ\tau_{T}=\mathbb{Q}_{a_{T};b_{T}}^{\eta,0\rightarrow\xi,T}$
and (cf. Lemma \ref{thm5.3-1})
\begin{equation}
h_{a_{T};b_{T}}(0,\eta,T,\xi)=h(0,\xi,T,\eta).\label{nb-01}
\end{equation}
Thanks to these dualities, we are able to rewrite the three terms
on the right-hand side of the representation \eqref{Rep-cas-01} for
$w^{i}$. Indeed the first term
\begin{align*}
R_{I}^{i}(\eta,T) & =\int_{D}\left(\int_{\varOmega}Q_{j}^{i}(\psi,T)1_{\left\{ \zeta_{D}(\psi)>T\right\} }\mathbb{Q}_{a_{T};b_{T}}^{\eta,0\rightarrow\xi,T}(\textrm{d}\psi)\right)w_{0}^{j}(\xi)h_{a_{T};b_{T}}(0,\eta,T,\xi)\textrm{d}\xi\\
 & =\int_{D}\left(\int_{\varOmega}\tilde{Q}_{j}^{i}(\psi,T;0)1_{\left\{ \zeta_{D}(\tau_{T}\psi)>T\right\} }\mathbb{Q}_{a;-b}^{\xi,0\rightarrow\eta,T}(\textrm{d}\psi)\right)w_{0}^{j}(\xi)h(0,\xi,T,\eta)\textrm{d}\xi.
\end{align*}
We notice that, if $\xi,\eta\in D$, then under the conditional law
$\mathbb{Q}^{\xi,0}\left[\left.\textrm{d}\psi\right|\psi(T)=\eta\right]$,
$\zeta_{D}(\tau_{T}\psi)>T$ is equivalent to that $\psi(T-t)\in D$
for all $t\in[0,T]$, which in turn is equivalent to that $\psi(t)\in D$.
While the last is equivalent to that $\zeta_{D}(\psi)>T$. Therefore
\begin{equation}
R_{I}^{i}(\eta,T)=\int_{D}\left(\int_{\varOmega}\tilde{Q}_{j}^{i}(\psi,T;0)1_{\left\{ \zeta_{D}(\psi)>T\right\} }\mathbb{Q}_{a;-b}^{\xi,0\rightarrow\eta,T}(\textrm{d}\psi)\right)w_{0}^{j}(\xi)h(0,\xi,T,\eta)\textrm{d}\xi.\label{R-I-01}
\end{equation}

To handle the second term which arises from the inhomogeneous boundary
data $\beta(x)$. By using the conditional law we may rewrite
\begin{align*}
R_{B}^{i}(\eta,T) & =\int_{\mathbb{R}^{d}}\int_{\varOmega}Q_{j}^{i}(\psi,\zeta_{D}(\psi))\beta^{j}(\psi(\zeta_{D}(\psi)))1_{\left\{ \zeta_{D}(\psi)\leq T\right\} }\mathbb{Q}_{a_{T};b_{T}}^{\eta,0\rightarrow\xi,T}(\textrm{d}\psi)h(0,\xi,T,\eta)\textrm{d}\xi\\
 & =\int_{\mathbb{R}^{d}}\int_{\varOmega}Q_{j}^{i}(\tau_{T}\psi,\zeta_{D}(\tau_{T}\psi))\beta^{j}(\tau_{T}\psi(\zeta_{D}(\tau_{T}\psi)))1_{\left\{ \zeta_{D}(\tau_{T}\psi)\leq T\right\} }\mathbb{Q}_{a;-b}^{\xi,0\rightarrow\eta,T}(\textrm{d}\psi)h(0,\xi,T,\eta)\textrm{d}\xi,
\end{align*}
where the complication arises due to the integral against the variable
$\xi$ takes place over the whole space $\mathbb{R}^{d}$. Observe
that $\zeta_{D}(\tau_{T}\psi)\leq T$ if and only if there is $t_{0}\in[0,T]$
such that $\psi(t_{0})\in\partial D$, which is therefore equivalent
to that $\zeta_{D}(\psi)\leq T$, and
\begin{align*}
\zeta_{D}(\tau_{T}\psi) & =\inf\left\{ t\geq0:\psi(T-t)\in\partial D\right\} \\
 & =\inf\left\{ T-s\geq0:\psi(s)\in\partial D\right\} \\
 & =T-\sup\left\{ s:0\leq s\leq T\textrm{ s.t. }\psi(s)\in\partial D\right\} .
\end{align*}
Therefore $Q_{j}^{i}(\tau_{T}\psi,\zeta_{D}(\tau_{T}\psi))=\tilde{Q}_{j}^{i}(\psi,\lambda_{T,\partial D}(\psi))$
on $\left\{ \zeta_{D}(\psi)\leq T\right\} $ and
\[
\beta^{j}(\tau_{T}\psi(\zeta_{D}(\tau_{T}\psi)))1_{\left\{ \zeta_{D}(\tau_{T}\psi)\leq T\right\} }=\beta^{j}(\psi(\lambda_{T,\partial D}(\psi)))1_{\left\{ \zeta_{D}(\psi)\leq T\right\} }.
\]
By using these relations we can rewrite the boundary term 
\begin{equation}
R_{B}^{i}(\eta,T)=\int_{\mathbb{R}^{d}}\left(\int_{\varOmega}\tilde{Q}_{j}^{i}(\psi,\lambda_{T,\partial D}(\psi))\beta^{j}(\psi(\lambda_{T,\partial D}(\psi)))1_{\left\{ \zeta_{D}(\psi)\leq T\right\} }\mathbb{Q}_{a;-b}^{\xi,0\rightarrow\eta,T}(\textrm{d}\psi)\right)h(0,\xi,T,\eta)\textrm{d}\xi.\label{R-B-t01}
\end{equation}
Finally let us consider the third term arising from the inhomogeneous
term in the parabolic system. In fact, by using duality we may rewrite
\begin{align*}
R_{N}^{i}(\eta,T) & =\int_{0}^{T}\int_{\mathbb{R}^{d}}\int_{\varOmega}Q_{j}^{i}(\psi,s)g^{j}(\psi(s),T-s)1_{\left\{ s<\zeta_{D}(\psi)\right\} }\mathbb{Q}_{a_{T};b_{T}}^{\eta,0\rightarrow\xi,T}(\textrm{d}\psi)h_{a_{T};b_{T}}(0,\eta,T,\xi)\textrm{d}\xi\textrm{d}s\\
 & =\int_{0}^{T}\int_{\mathbb{R}^{d}}\int_{\varOmega}Q_{j}^{i}(\tau_{T}\psi,s)g^{j}(\tau_{T}\psi(s),T-s)1_{\left\{ s<\zeta_{D}(\tau_{T}\psi)\right\} }\mathbb{Q}_{a;-b}^{\xi,0\rightarrow\eta,T}(\textrm{d}\psi)h(0,\xi,T,\eta)\textrm{d}\xi\textrm{d}s\\
 & =\int_{0}^{T}\int_{\mathbb{R}^{d}}\int_{\varOmega}\tilde{Q}_{j}^{i}(\psi,s)g^{j}(\psi(s),s)1_{\left\{ \zeta_{D}(\theta_{s}\psi)>T-s\right\} }\mathbb{Q}_{a;-b}^{\xi,0\rightarrow\eta,T}(\textrm{d}\psi)h(0,\xi,T,\eta)\textrm{d}\xi\textrm{d}s.
\end{align*}
Putting these equations together we deduce the functional integration
representation.
\end{proof}
In particular we have the following forward Feynman-Kac formula.
\begin{thm}
\label{thm7.2-1} Suppose $\nabla\cdot b=0$ on $\mathbb{R}^{d}$
in the distribution sense. Let $w^{i}(x,t)$ be the solution to Cauchy's
initial problem of the parabolic system:
\begin{equation}
\left(\mathscr{L}_{a;b}-\frac{\partial}{\partial t}\right)w^{j}(x,t)+\sum_{k=1}^{n}q_{k}^{j}(x,t)w^{k}(x,t)=0\quad\textrm{ in }D\times[0,T]\label{sy-01-1}
\end{equation}
subject to the initial and Dirichlet boundary conditions:
\begin{equation}
w^{j}(x,0)=w_{0}^{j}(x)\textrm{ for }x\in D,\textrm{ and }w^{j}(x,t)=0\textrm{ for }x\in\partial D\textrm{, }t>0\label{int-bd-c3-1}
\end{equation}
where $j=1,\cdots,n$. Then
\[
w^{i}(\eta,T)=\int_{D}\left(\int_{\varOmega}\tilde{Q}_{j}^{i}(\psi,T;0)1_{\left\{ \zeta_{D}(\psi)>T\right\} }\mathbb{Q}_{a;-b}^{\xi,0\rightarrow\eta,T}(\textrm{d}\psi)\right)w_{0}^{j}(\xi)h_{a;-b}(0,\xi,T,\eta)\textrm{d}\xi
\]
 for every $\eta\in D$.
\end{thm}

\begin{rem}
We would like to emphasize the assumptions on $a^{ij}(x,t)$ and $b^{i}(x,t)$,
both are defined for all $x\in\mathbb{R}^{d}$ and $t\geq0$. In order
to ensure the previous functional integration representation to be
valid, we assume that $b^{i}(x,t)$ is bounded (this condition can
be weaken) and Borel measurable, but the most crucial assumption is
that $b(x,t)$ is divergence free on $\mathbb{R}^{d}$ (not only on
$D$ !) in the distribution sense. The probability measure used in
the representation is the distribution associated with the differential
operator
\[
\mathscr{L}_{a;-b}=\nu\sum_{i,j=1}^{d}\frac{\partial}{\partial x^{j}}a^{ij}(x,t)\frac{\partial}{\partial x^{i}}-\sum_{i=1}^{d}b^{i}(x,t)\frac{\partial}{\partial x^{i}}
\]
which is the formal adjoint operator of $\mathscr{L}_{a;b}$.
\end{rem}

\section{Navier-Stokes equations}

In this section we apply the forward Feynman-Kac formula to derive
a stochastic representation for solutions of the Navier-Stokes equations
in domains via their vortex dynamics, cf. \citep{Cottet and Koumoutsakos 2000,Majda and Bertozzi 2002}.
We begin with the following elementary fact.
\begin{lem}
\label{lem9.1}Suppose $b(x)$ is a $C^{1}$-vector field in $\overline{D}$
, and $\nabla\cdot b(x)=0$ for all $x\in D$. Suppose $b(x)=0$ for
$x\in\partial D$. Extend $b(x)$ to all $x\in\mathbb{R}^{d}$ by
setting $b(x)=0$ for $x\notin\overline{D}$. Then $b(x)$ is divergence-free
in distribution sense on $\mathbb{R}^{d}$.
\end{lem}

\begin{proof}
According to assumptions, $\partial D$ is a smooth manifold of $d-1$
dimensions, so has zero Lebesgue measure. Therefore $\nabla\cdot b=0$
a.e. on $\mathbb{R}^{d}$. Suppose $\varphi$ is a smooth function
on $\mathbb{R}^{d}$ with a compact support. Then
\[
\nabla\cdot(\varphi b)=\nabla\varphi\cdot b+\varphi\nabla\cdot b=\nabla\varphi\cdot b\quad\textrm{ a.e. }
\]
Hence
\begin{align*}
\int_{\mathbb{R}^{d}}\nabla\varphi\cdot b & =\int_{\mathbb{R}^{d}}\nabla\cdot(\varphi b)=\int_{\partial D}\varphi b\cdot\boldsymbol{\nu}+\int_{\partial D^{c}}\varphi b\cdot\boldsymbol{\nu}\\
 & =0
\end{align*}
where the last equality follows from the assumption that $b(x)=0$
for $x\in\partial D$. Therefore $\nabla\cdot b=0$ on $\mathbb{R}^{d}$
in the distribution sense. 
\end{proof}
Recall that the velocity $u(x,t)$ and the pressure $P(x,t)$ of an
incompressible fluid flow constrained in $D$ are solutions of the
Navier-Stokes equations
\begin{equation}
\frac{\partial}{\partial t}u+(u\cdot\nabla)u-\nu\Delta u-\nabla P=0\quad\textrm{ in }D\times[0,\infty)\label{NS-a1}
\end{equation}
and
\begin{equation}
\nabla\cdot u=0\quad\textrm{ in }D\times[0,\infty),\label{NS-2}
\end{equation}
subject to the no-slip condition
\begin{equation}
u(x,t)=0\quad\textrm{ for }x\in\partial D\textrm{ and }t\geq0.\label{NS-3}
\end{equation}
Suppose the initial data $u_{0}\in C^{\infty}(\overline{D})$. Then,
in dimension two, $u(x,t)$ remains smooth in $(x,t)$ up to the boundary,
while in dimension three, the regularity of $u(x,t)$ remains open.
Let $\omega=\nabla\wedge u$ be the vorticity and $\omega_{0}=\nabla\wedge u_{0}$.
It can be verified by a simple calculation that the boundary vorticity
can be related to the shearing stress $\tau_{ij}=\nu\left(\frac{\partial u^{i}}{\partial x^{j}}+\frac{\partial u^{j}}{\partial x^{i}}\right)$
applied immediately to the boundary surface. In fact it can be demonstrated
that the normal part $\omega^{\perp}$ of the boundary vorticity $\left.\omega\right|_{\partial D}$,
vanishes identically along the boundary surface, while its tangential
vorticity $\omega^{\parallel}$ along the boundary coincides with
the normal shearing stress $\tau^{\perp}$ up to a numerical factor
$\nu^{-1}$. The calculation of the boundary vorticity is an important
problem which will not be discussed in this paper in detail. However
let us point out that the normal stress applied immediately to the
wall is a fluid dynamical quantity to be measured or to be controlled,
and can be calculated approximately by using boundary layer equation,
cf. \citep[Chapter 6]{Schlichting9th-2017}.

\subsection{Two dimensional flows}

Let us first establish a representation for 2D flows. In dimension
two, the vorticity $\omega$ can be identified with the scalar function
$\frac{\partial}{\partial x^{1}}u^{2}-\frac{\partial}{\partial x^{2}}u^{1}$
and $\omega$ is a solution to the vorticity transport equation
\begin{equation}
\frac{\partial}{\partial t}\omega+(u\cdot\nabla)\omega-\nu\Delta\omega=0\quad\textrm{ in }D\times[0,\infty),\label{vor-01}
\end{equation}
where the boundary value of $\omega$ along the wall $\partial D$
may be identified with the stress of the fluid flow immediately injected
to the wall, cf. \citep{Schlichting9th-2017,Liu and Weinan2000,Weinan and Liu 1997}.
Let us denote the stress along the wall by $\sigma$, whose explicit
expression need to be calculated in terms of the geometry of $\partial D$
as well, so they are must be treated case by case. 

Since $\nabla\cdot u=0$, the velocity field may be recovered in terms
of $\omega$ by solving the Poisson equations
\begin{equation}
\Delta u^{1}=-\frac{\partial\omega}{\partial x^{2}},\quad\Delta u^{2}=\frac{\partial\omega}{\partial x^{1}}\label{lap-02}
\end{equation}
subject to the Dirichlet boundary condition $u^{1}(x,t)=u^{2}(x,t)=0$
for $x\in\partial D$. Hence according to Green formula 
\begin{equation}
u^{i}(x,t)=\int_{D}K^{i}(x,\eta)\omega(\eta,t)\textrm{d}\eta=\int_{\mathbb{R}^{2}}K_{D}^{i}(x,\eta)\omega(\eta,t)\textrm{d}\eta\label{lap-03}
\end{equation}
where 
\begin{equation}
K_{D}^{i}(x,\eta)=1_{D}(\eta)K^{i}(x,\eta)\label{Ki-ee1}
\end{equation}
and the integral kernel $K^{i}$ depend on the region $D$ only. 
\begin{thm}
\label{thm9.2} Let $u$ be extended to be a vector field on $\mathbb{R}^{2}\times[0,\infty)$
such that $u(\cdot,t)$ is divergence free in the distribution sense
on the whole plane $\mathbb{R}^{2}$. Let $X(\xi,t)$ (for $\xi\in\mathbb{R}^{2}$
and $t\geq0$) be the solution to the stochastic differential equation:
\begin{equation}
\textrm{d}X(t)=u(X(t),t)\textrm{d}t+\sqrt{2\nu}\textrm{d}B(t),\quad X(0)=\xi.\label{s-re01}
\end{equation}
Then
\begin{align}
u(x,t) & =\int_{D}\omega_{0}(\xi)\mathbb{E}\left[K_{D}(x,X(\xi,t))J_{1}(\xi,X(\xi,t),t)\right]\textrm{d}\xi\nonumber \\
 & +\int_{\mathbb{R}^{2}}\mathbb{E}\left[K_{D}^{i}(x,X(\xi,t))J_{2}(\xi,X(\xi,t),t)\right]\textrm{d}\xi\label{2D-rep-01}
\end{align}
for every $x\in D$ and $t>0$, where 
\begin{equation}
J_{1}(\xi,\eta,t)=\mathbb{P}\left[\left.\zeta_{D}(X(\xi,\cdot))>t\right|X(\xi,t)=\eta\right]\label{J1-e2}
\end{equation}
and
\begin{equation}
J_{2}(\xi,\eta,t)=\mathbb{E}\left[\left.\sigma\left(X(\xi,\lambda_{t,\partial D}(X(\xi,\cdot))\right)1_{\left\{ \zeta_{D}(X(\xi,\cdot))\leq t\right\} }\right|X(\xi,t)=\eta\right]\label{J2-e2}
\end{equation}
for any $\xi,\eta\in\mathbb{R}^{2}$ and $t>0$.
\end{thm}

\begin{proof}
The vorticity transport equation may be formulated in terms of $\mathscr{L}_{-u}$,
that is,
\begin{equation}
\left(\frac{\partial}{\partial t}-\mathscr{L}_{-u}\right)\omega=0\quad\textrm{ in }D\times[0,\infty).\label{vor-02}
\end{equation}
 Since $u(\cdot,t)$ is divergence free in the distribution sense
for every $t$ on $\mathbb{R}^{2}$, therefore we may apply Theorem
\ref{thm7.2-1} to $\omega$, to obtain
\begin{align}
\omega(\eta,t) & =\int_{D}\omega_{0}(\xi)\mathbb{P}^{\xi,0\rightarrow\eta,t}\left[\zeta_{D}(\psi)>t\right]h(0,\xi,t,\eta)\textrm{d}\xi\nonumber \\
 & +\int_{\mathbb{R}^{2}}\left(\int_{\varOmega}\sigma(\psi(\lambda_{t,\partial D}(\psi)))1_{\left\{ \zeta_{D}(\psi)\leq t\right\} }\mathbb{P}^{\xi,0\rightarrow\eta,t}(\textrm{d}\psi)\right)h(0,\xi,t,\eta)\textrm{d}\xi\label{vort-03}
\end{align}
for $\eta\in D$ and $t>0$, where $\mathbb{P}^{\xi,0\rightarrow\eta,T}$
denotes the conditional distribution of the $\mathscr{L}_{u}$-diffusion
started from $\xi$ at time zero given $\psi(t)=\eta$, and $h(\tau,\xi,t,\eta)$
is the transition probability density function of the $\mathscr{L}_{u}$-diffusion.
Since the distribution of $X(\xi,t)$ is exactly $\mathbb{P}^{\xi,0}$,
the conclusion therefore follows immediately.
\end{proof}

\subsection{Three dimensional flows}

Next we consider an incompressible fluid flow with its velocity $u=(u^{1},u^{2},u^{3})$,
constrained in a region $D\subset\mathbb{R}^{3}$. Therefore $u(x,t)$
is a solution to the 3D Navier-Stokes equations (\ref{NS-a1}, \ref{NS-2},
\ref{NS-3}). The vorticity $\omega^{i}=\varepsilon^{ijk}\frac{\partial}{\partial x^{j}}u^{k}$
are solutions to the 3D vorticity transport equations
\begin{equation}
\frac{\partial}{\partial t}\omega^{i}+(u\cdot\nabla)\omega^{i}-\nu\Delta\omega^{i}-\frac{\partial u^{i}}{\partial x^{j}}\omega^{j}=0\quad\textrm{ in }D\times[0,\infty),\label{vort05}
\end{equation}
where $\nu>0$ is the viscosity constant as usual, which can be written
as
\begin{equation}
\left(\mathscr{L}_{-u}-\frac{\partial}{\partial t}\right)\omega^{i}+S_{j}^{i}\omega^{j}=0,\label{vort04}
\end{equation}
where 
\[
S_{j}^{i}=\frac{1}{2}\left(\frac{\partial u^{i}}{\partial x^{j}}+\frac{\partial u^{j}}{\partial x^{i}}\right)
\]
is the symmetric stress tensor. The boundary value of $\omega$ along
the wall $\partial D$ can be again identified with the stress tensor,
denoted by $\beta$. Since $\nabla\cdot u=0$ in $D$ so that $\Delta u=-\nabla\times\omega$
in $D$ and $u$ satisfies the Dirichlet boundary condition along
$\partial D$. Therefore
\begin{equation}
u(x,t)=-\int_{D}H(x,\eta)\nabla\times\omega(\eta,t)\textrm{d}\eta,\label{3d-BS-01}
\end{equation}
where $H(x,y)$ is the Green function of $D$. By integration by parts
we obtain
\begin{align}
u(x,t) & =\int_{D}K(x,\eta)\times\omega(\eta,t)\textrm{d}\eta\nonumber \\
 & =\int_{\mathbb{R}^{3}}K_{D}(x,\eta)\times\omega(\eta,t)\textrm{d}\eta\label{u-a01}
\end{align}
where $K_{D}(x,\eta)=1_{D}(\eta)K(x,\eta)$. That is
\begin{equation}
u^{i}(x,t)=\int_{D}\varepsilon^{ilk}K^{l}(x,\eta)\omega^{k}(\eta,t)\textrm{d}\eta.\label{u-a02}
\end{equation}
It follows that 
\begin{align}
S_{j}^{i}(x,t) & =\int_{D}K_{j}^{i,k}(x,\eta)\omega^{k}(\eta,t)\textrm{d}\eta\nonumber \\
 & =\int_{D}K_{D;j}^{i,k}(x,\eta)\omega^{k}(\eta,t)\textrm{d}\eta,\label{q-a01}
\end{align}
where
\begin{equation}
K_{j}^{i,k}(x,\eta)=\varepsilon^{ilk}\frac{\partial}{\partial x^{j}}K^{l}(x,\eta)+\varepsilon^{jlk}\frac{\partial}{\partial x^{i}}K^{l}(x,\eta)\label{kernel-K01}
\end{equation}
and
\begin{equation}
K_{D;j}^{i,k}(x,\eta)=1_{D}(\eta)\left(\varepsilon^{ilk}\frac{\partial}{\partial x^{j}}K^{l}(x,\eta)+\varepsilon^{jlk}\frac{\partial}{\partial x^{i}}K^{l}(x,\eta)\right).\label{kernel-K02}
\end{equation}
We notice that the Green function $G$, the integral kernels $K$
and $K_{D}$ are determined solely by the domain $D\subset\mathbb{R}^{3}$. 

Let $u(x,t)$ be extended to be a divergence free (in the distribution
sense) vector field on $\mathbb{R}^{3}$, and $X(\xi;t)$ and $\tilde{Q}_{j}^{i}(x,t)$
are the solutions to the stochastic differential equations 
\begin{equation}
\textrm{d}X(\xi;t)=u(X(\xi;t),t)\textrm{d}t+\sqrt{2\nu}\textrm{d}B(t),\quad X(\xi;0)=\xi\label{sde-3d-01}
\end{equation}
and 
\begin{equation}
\frac{\textrm{d}}{\textrm{d}s}\tilde{Q}_{j}^{i}(\xi,t;s)=-\tilde{Q}_{k}^{i}(\xi,t;s)1_{D}(X(\xi;s))S_{j}^{k}(X(\xi;s),s),\quad\tilde{Q}_{j}^{i}(\xi,t;t)=\delta_{j}^{i}\label{sde-3d-02}
\end{equation}
for $t>0$ and $s\geq0$. 
\begin{thm}
\label{thm9.3}Suppose $u(x,t)$ is smooth and bounded on $\overline{D}\times[0,T]$,
then 
\begin{align}
u^{k}(x,t) & =\int_{D}\omega_{0}^{l}(\xi)\mathbb{E}\left[\varepsilon^{ijk}K_{D}^{j}(x,X(\xi,t))J_{l}^{i}(\xi,X(\xi,t),t)\right]\textrm{d}\xi\nonumber \\
 & +\int_{\mathbb{R}^{3}}\mathbb{E}\left[\varepsilon^{ijk}K_{D}^{j}(x,X(\xi,t))B^{i}(\xi,X(\xi,t),t)\right]\textrm{d}\xi\label{rep-3d-1}
\end{align}
and
\begin{align}
S_{j}^{i}(x,t) & =\int_{D}\omega_{0}^{l}(\xi)\mathbb{E}\left[K_{D;j}^{i,k}(x,\eta)J_{l}^{k}(\xi,X(\xi,t),t)\right]\textrm{d}\xi\nonumber \\
 & +\int_{\mathbb{R}^{3}}\mathbb{E}\left[K_{D;j}^{i,k}(x,\eta)B^{k}(\xi,X(\xi,t),t)\right]\textrm{d}\xi\label{rep-002}
\end{align}
for all $x\in D$ and $t\in(0,T]$, where $k=1,2,3$, 
\begin{equation}
J_{j}^{i}(\xi,\eta,t)=\mathbb{E}\left[\left.\tilde{Q}_{j}^{i}(\xi,t;0)1_{\left\{ \zeta_{D}(X(\xi,\cdot))>t\right\} }\right|X(\xi,t)=\eta\right]\label{3dij}
\end{equation}
and
\begin{equation}
B^{i}(\xi,\eta,t)=\mathbb{E}\left[\left.\tilde{Q}_{j}^{i}(\xi,t;\lambda_{T,\partial D}(X(\xi,\cdot)))1_{\left\{ \zeta_{D}(X(\xi,\cdot))\leq t\right\} }\beta^{j}(X(\xi,\lambda_{T,\partial D}(X(\xi,\cdot)))\right|X(\xi,t)=\eta\right].\label{3d-b}
\end{equation}
\end{thm}

\begin{proof}
The proof is similar to that of two dimensional case. According to
Theorem \ref{thm7.3}
\begin{align*}
\omega^{i}(\eta,T) & =\int_{D}J_{l}^{i}(\xi,\eta,t)\omega_{0}^{l}(\xi)h_{u}(0,\xi,T,\eta)\textrm{d}\xi\\
 & +\int_{\mathbb{R}^{d}}B^{i}(\xi,\eta,t)h_{u}(0,\xi,T,\eta)\textrm{d}\xi
\end{align*}
and the representation formula follows from (\ref{3d-BS-01}) and
the Fubini theorem immediately.
\end{proof}

\end{document}